\documentclass{article}

    \PassOptionsToPackage{numbers, sort&compress}{natbib}

\usepackage[preprint]{neurips_2026}

\usepackage[utf8]{inputenc} 
\usepackage[T1]{fontenc}    
\usepackage{hyperref}       
\usepackage{url}            
\usepackage{booktabs}       
\usepackage{amsfonts}       
\usepackage{nicefrac}       
\usepackage{microtype}      
\usepackage{xcolor}         

\usepackage{subcaption}

\title{Proportionality in Ranking Compression}

\author{%
  Daniel Halpern \\
  Google Research 
  \And
  Xinyu Liu \\
  U. of Southern California
  \And
  Evi Micha \\
  U. of Southern California
  \AND
  Yu Peng Ng\\
  National University of Singapore
  \And
  Warut Suksompong\\
  National University of Singapore
}

\usepackage{amsmath}
\usepackage{amsthm}
\usepackage{cleveref}
\usepackage{tikz}
\usepackage{float}
\usepackage{caption}

\newtheorem{theorem}{Theorem}

\newtheorem{proposition}{Proposition}
\newtheorem{lemma}{Lemma}

\theoremstyle{definition}
\newtheorem{definition}{Definition}

\usepackage{enumitem}
\usepackage{algorithm}
\usepackage[noend]{algpseudocode}
\usepackage{amssymb}

\newcommand{\sto}{\xrightarrow{\text{s}}}
\newcommand{\wto}{\xrightarrow{\text{w}}}
\newcommand{\nto}[1]{\xrightarrow{#1}}

\newcommand{\Top}[2]{\mathrm{Top}_{#1}(#2)}

\newcommand{\ceil}[1]{\left \lceil #1 \right \rceil}
\newcommand{\floor}[1]{\left \lfloor #1 \right \rfloor}

\allowdisplaybreaks

\usepackage{color-edits}
\usepackage[dvipsnames]{xcolor}

\begin{document}

\maketitle

\begin{abstract}
Motivated by applications in reinforcement learning from human feedback, we consider a setting where a given collection of rankings needs to be compressed into a smaller collection of rankings.
This can also be seen as a combination of two fundamental social choice frameworks: ranking aggregation and multiwinner voting.
We propose three proportionality notions that capture different types of representation in this setting: positional proportionality, pairwise proportionality, and proportionality for solid coalitions (PSC).
On the one hand, we show that positional proportionality and PSC are always satisfiable for any input rankings and target compression size, and a desired output can be found in polynomial time. 
On the other hand, we prove that pairwise proportionality cannot be satisfied in general, but can nevertheless be attained when the input rankings are single-peaked or single-crossing.
\end{abstract}

\section{Introduction}
\label{sec:intro}

The goal of AI alignment is to build systems that act consistently with human values. For large language models, the predominant method is \emph{reinforcement learning from human feedback} (RLHF)~\citep{ouyang2022training}. 
This process typically consists of three stages: collect user preferences over possible outputs, fit a reward function to best match these choices, and fine-tune the policy to optimize the learned signal.
However, standard RLHF implicitly assumes a shared objective among evaluators. 
When ``correct'' behavior naturally diverges across cultures or individuals with different backgrounds~\citep{prabhakaran2022cultural,yuan2024cultural,anwar2024foundational,kirk2024prism,zhang2025cultivating}, aggregating these varied viewpoints into a single target marginalizes minority perspectives and leads to \emph{preference collapse}~\citep{xiao2024algorithmic}, an extreme lack of diversity in outputs~\citep{khalifa2020distributional,perez2022red,santurkar2023whose,kirk2024understanding,shypula2025evaluating,lake2025from}. 

In response to this challenge, the subfield of \emph{pluralistic alignment} has emerged, with an explicit goal of preserving and utilizing diverse belief systems~\citep{roadmap}. 
One practical approach is to learn multiple reward functions from feedback data that collectively cover these diverse views~\citep{coste2024reward,chakraborty2024maxmin,poddar2024personalizing,chen2025pal,HalpernMiPr25}.
For example, we could train a distinct model for each reward and deploy them as an ensemble~\citep{roadmap}.  Alternatively, these rewards or their resulting learned models can act as a ``jury'' that effectively votes on the final output~\citep{gordon2022jury}.
In an idealized world, we might learn a unique reward function for every individual in society. 
While this would cover the full spectrum of views, scaling it to millions (or billions) of personalized rewards is computationally prohibitive.  
Therefore, to implement any of these ensemble methods, it is necessary to compress the global preferences into a relatively small set of reward functions while preserving their underlying diversity.

To formalize this approach, we turn to the field of \emph{social choice theory}, which studies how to aggregate individual preferences into a collective outcome~\citep{ArrowSeSu02,BrandtCoEn16}. 
In our setting, human annotators naturally provide ranked feedback, and a reward function can be evaluated by how it orders candidate outputs~\citep{ge2024axioms}. Thus, at an abstract level, our problem reduces to the following: given an input profile of $n$ rankings, construct $k \le  n$ output rankings that optimally represent the input.\footnote{Admittedly, there are some practical gaps between this theoretical abstraction and traditional RLHF frameworks. For instance, annotators sometimes provide only sparse pairwise comparisons, making it difficult to infer complete preferences (though recent datasets are beginning to collect substantially richer preference data per user~\citep{zhang2025cultivating}). Nevertheless, our primary focus is to establish the theoretical limits of ranking compression; we discuss possible approaches to bridge this divide in \Cref{sec:future-directions}.}
This formalization casts our problem as an intersection of two classic social choice frameworks: \emph{multiwinner voting}, which aims to choose a representative set of candidates, and \emph{rank aggregation}, which concerns selecting a single winning ranking (the setting of Arrow's famous impossibility theorem~\citep{Arrow50}). 
By contrast, our setting---which we will refer to as \emph{ranking compression}---requires generating a \emph{representative collection of rankings}.
Beyond RLHF, the same problem appears whenever a large collection of preferences must be represented by a small number of viewpoints. For instance, members of a population may rank a common set of policy or project proposals, while a committee or deliberative body has room for only $k$ representative preference orders. 
The objective of ranking compression is to construct these $k$ rankings so that significant patterns in the $n$~original rankings remain proportionally represented.\footnote{Note that we do not require the output rankings to be distinct.
An interpretation of these rankings is as the votes of a representative body, and distinct representatives may have identical preferences.
For example, if every voter in the population has the same preference, a representative body of size~$k$ would naturally consist of $k$~distinct representatives who all share the same consensus ranking.
}

Despite the seemingly simple setup, formalizing the concept of ``optimal representation'' in this context is far from trivial. 
Indeed, compressing a large preference profile into a smaller set is inherently lossy, thereby requiring a formal definition of what it means to faithfully reflect the original preferences. 
Since there is no single canonical choice---and the ideal tradeoff inevitably depends on the downstream task---we instead adopt an axiomatic approach, seeking to establish precisely which properties can or cannot be mathematically preserved under compression.

\paragraph{Our results.}

We introduce three natural axioms that capture distinct facets of proportional representation in ranking compression. 
\begin{itemize}[leftmargin=*, nosep]
    \item \emph{Positional proportionality}: The fraction of input rankings that place a candidate at or above a given rank must be preserved in the output.
    \item \emph{Pairwise proportionality}: The pairwise comparison between candidates in the input instance must be preserved in the output.
    \item \emph{Proportionality for solid coalitions (PSC)}: If a group of input rankings ranks a set of candidates above all remaining candidates, this must be reflected in the output. 
    This extends the well-known multiwinner voting axiom (with the same name) to our setting~\citep{Dummett84,AzizLe20,BrillPe23}.
\end{itemize}

Our findings offer a holistic analysis of these axioms, demonstrating both existence guarantees (all with accompanying efficient algorithms) as well as impossibility results. 

On the positive side, we show that positional proportionality is always achievable for any given input instance and target compression size (\Cref{thm:positional-prop}).
We also prove that the same holds for PSC (\Cref{thm:PSC}).
Our proofs for these results leverage various combinatorial and algorithmic tools, including Hall's marriage theorem and network flow.

In contrast, we establish an impossibility for pairwise proportionality, by proving that it cannot necessarily be satisfied for any $k\ge 2$ and $n \ge k+2$ (\Cref{thm:pairwise-prop}).
Our counterexamples rely on intricate constructions of constraint graphs and arguments about when the corresponding constraints can be fulfilled by a set of rankings.
Despite this negative result, we show that pairwise proportionality can be guaranteed under \emph{single-peaked} or \emph{single-crossing} preferences, two well-studied restricted domains in social choice (\Cref{thm:single-peaked-simul-prop,thm:pairwise-prop-single-crossing}). 
Moreover, for single-peaked instances we can ensure that the output simultaneously satisfies the other two proportionality notions and preserves single-peakedness, while for single-crossing instances we can select the output from the input rankings themselves in such a way that it also satisfies PSC.

\paragraph{Related work.}

In terms of problem formulation, our framework shares its basic premise with a ``subset selection'' problem focused on choosing a representative set of LLM benchmarks from a large pool~\citep{ProcacciaScWa25}. 
That prior work similarly requires outputting $k$ rankings based on an input of $n$ rankings, and defines an approximate version of positional proportionality. 
However, it crucially demands the output rankings to be drawn directly from the input rankings. 
By 
removing this restriction, our setting offers a much larger solution space, thereby necessitating an entirely different set of combinatorial tools.
In fact, as we demonstrate in \Cref{prop:output-from-input} and Appendix~\ref{app:single-crossing-output-from-input}, most of our positive results cease to hold under the output-from-input restriction.

Conceptually, ranking compression shares similarities with the clustering of rankings \citep{viappiani2015characterization,Faliszewski2023Diversity,faliszewski2026diversity}. 
The general objective of clustering is to partition a preference profile into $k$ disjoint sets, where the $k$ ``central'' rankings derived from these clusters can be interpreted as a meaningful compression of the profile. 
However, existing clustering frameworks differ from our approach in several fundamental ways. 
For example, clusters are typically allowed to vary in size, which suggests that the $k$ representative rankings should be interpreted as weighted. 
Additionally, a core tenet of traditional clustering is that each voter is explicitly assigned to a specific cluster center. In contrast, our work focuses on achieving a global axiomatic representation without requiring any such individual assignments.

Within the clustering literature, \citet{viappiani2015characterization} examines a standard distance-based clustering method designed to find $k$ rankings that minimize the sum of distances. 
Recent research has also utilized clustering to measure election polarization \citep{Faliszewski2023Diversity} and investigated the computational complexity of identifying optimal cluster centers, such as $k$-Kemeny scores \citep{faliszewski2026diversity}. Further afield are notions of proportional clustering in general metric spaces~\citep{chhabra2021overview,chen2019proportionally,micha2020proportionally,aziz2024proportionally,kellerhals2024proportional}. 
While these notions do not normally involve rankings, they can in principle be extended to our setting using ranking-specific metrics such as the Kendall-tau distance or Spearman's footrule.  

Recent work in social choice theory has addressed proportionality in ranking aggregation \citep{LedererPeWa24,Lederer26}.
This line of work assumes a single output ranking and aims for proportionality to be reflected within that ranking itself, which makes the challenges significantly different from ours.
In particular, with $k = 1$ output ranking, satisfying all three of our proportionality notions is trivial: if all input rankings are the same, return that ranking; otherwise, return an arbitrary ranking.
A related model involves aggregating approval votes into a proportional ranking \citep{SkowronLaBr16}.

Finally, while notions of compression exist within social choice theory, they differ significantly from our approach. 
For instance, \emph{compilation complexity}~\citep{maudet2009compiling,xia2010compilation,Chevaleyre2} measures the minimum number of bits needed to encode the votes of a subpopulation without altering the final election outcome. 
Another related notion is the \emph{communication-distortion tradeoff}~\citep{mandal2019efficient,mandal2020optimal}, which concerns how to efficiently compress real-valued utility functions into fixed-size ballots in order to maximize social welfare. 

\section{Preliminaries}
\label{sec:prelim}

For any positive integers~$z\le z'$, let $[z] = \{1,\dots,z\}$, $[z]_0 = \{0,1,\dots,z\}$, and $[z,z'] = \{z,{z+1},\dots,z'\}$.
Let $N = [n]$ be a set of $n$~voters and $C = [m]$ be a set of $m$~candidates, for some positive integers $n,m$.
Each voter $i\in N$ has a strict ranking $\pi_i = (c_{i,1},\dots,c_{i,m})$ of the candidates in~$C$; for $\ell\in[m]$, we say that $c_{i,\ell}$ is in the $\ell$-th rank. 
The goal is to compress these $n$ input rankings into an output  consisting of $k$ rankings, where $k\in[n]$ is a given target compression size.\footnote{While the case $k = n$ is uninteresting since one may simply output the $n$ input rankings, we allow it in our formal model.} 
These $k$ output rankings need not come from the $n$ input rankings, nor must they be distinct.
An \emph{instance} of the ranking compression problem is described by a tuple $(N, C, (\pi_i)_{i\in N}, k)$.

We now define three notions of proportionality for the output with respect to the input rankings.

\textbf{Positional proportionality:}
Given an instance, an output is said to satisfy \emph{positional proportionality} if the following holds:
For any $t\in [k]$ and $\ell \in [m]$, if a candidate appears in the top $\ell$ ranks in at least $t\cdot n/k$ input rankings, then it must appear in the top $\ell$ ranks in at least $t$ output rankings.

\textbf{Pairwise proportionality:}
Given an instance, an output is said to satisfy \emph{pairwise proportionality} if the following holds:
For any $t\in[k]$ and candidates $c,c'\in C$, if $c$ is ranked higher than $c'$ in at least $t\cdot n/k$ input rankings, then $c$ must be ranked higher than $c'$ in at least $t$ output rankings.

\textbf{PSC:} Given an instance, an output is said to satisfy \emph{proportionality for solid coalitions (PSC)} if the following holds:
For any $t\in[k]$ and $C'\subseteq C$, if at least $t\cdot n/k$ input rankings have all candidates from~$C'$ in the top $|C'|$ ranks (but may order the candidates from~$C'$ differently within the top $|C'|$ ranks), then at least $t$ output rankings must have all candidates from~$C'$ in the top $|C'|$ ranks.

Intuitively, pairwise proportionality requires the output to preserve the voters' local preferences: if a fraction of voters prefers candidate~$c$ to candidate~$c'$, roughly the same fraction of output rankings should rank $c$ above~$c'$.
PSC takes a different route: rather than preserving every pairwise comparison, it ensures that every sufficiently large group of voters that agree on the top candidates is given a proportional voice among the $k$~rankings. 
On the other hand, positional proportionality can be viewed as demanding proportional representation from the candidates' perspective: candidates that frequently occupy prominent positions in the input should retain a proportional presence in those positions in the output.

As one may expect, all three notions are incomparable.
To see this formally, consider an instance with $(n,m,k) = (4,3,2)$ where two voters have the ranking $(1,2,3)$, one voter has the ranking $(2,3,1)$, and one voter has the ranking $(3,1,2)$.
\begin{itemize}
\item The output $(1,3,2), (3,2,1)$ satisfies positional proportionality, but fails pairwise proportionality (for the pair $(2,3)$) and PSC (for the set $\{1,2\}$).
\item The output $(2,1,3), (3,1,2)$ satisfies pairwise proportionality, but fails positional proportionality (for candidate~$1$) and PSC (for the set $\{1\}$).
\item The output $(1,2,3), (2,1,3)$ satisfies PSC, but fails positional proportionality (for candidate~$3$) and pairwise proportionality (for the pair $(3,1)$).
\end{itemize}

The following proposition demonstrates that requiring the output rankings to come from the input rankings severely limits the potential for attaining proportionality.

\begin{proposition}
\label{prop:output-from-input}
Suppose that the output rankings must be drawn from the input rankings.
Then, there exists an instance such that none of positional proportionality, pairwise proportionality, and PSC can be attained.
\end{proposition}

The full proof of \Cref{prop:output-from-input} can be found in Appendix~\ref{app:output-from-input}; we provide a sketch here.
Consider an instance with $(n, m, k) = (4, 6, 2)$ and
\begin{align*}
\pi_1 &= (3, 4, 2, 5, 1, 6), \hspace{0.5em}
\pi_2 = (3, 4, 5, 2, 6, 1), \hspace{0.5em}
\pi_3 = (4, 3, 2, 5, 6, 1), \hspace{0.5em}
\pi_4 = (4, 3, 5, 2, 1, 6).
\end{align*}
For positional proportionality, observe that for each pair of input rankings, there exist $c\in C$ and $\ell\in[m]$ such that candidate $c$ appears in the top $\ell$ ranks \emph{only} within that pair.
For example:
\begin{itemize}
  \item $\pi_1,\pi_2$: candidate $3$ appears in the top $1$ rank.
  \item $\pi_1,\pi_3$: candidate $2$ appears in the top $3$ ranks.
  \item $\pi_1,\pi_4$: candidate $1$ appears in the top $5$ ranks.
\end{itemize}
Analogous statements can be made for the remaining three pairs.
Now, each pair imposes a positional constraint on the output that requires selecting at least one of those two rankings. Because no output of size two can intersect all six pairs, positional proportionality cannot be achieved. 
One can verify that PSC and pairwise proportionality are also unattainable using similar arguments.

\section{Positional proportionality}
\label{sec:positional-prop}

In this section, we consider positional proportionality.
We show that an output satisfying positional proportionality always exists regardless of the input rankings and target compression size. 

\begin{theorem}
\label{thm:positional-prop}
For any instance, there exists an output that satisfies positional proportionality.
Moreover, such an output can be computed in polynomial time.
\end{theorem}

To prove this theorem, we will need the following combinatorial lemma, which can be shown using Hall's marriage theorem.
Recall that an $a\times b$ table has $a$~rows and $b$~columns.

\begin{lemma}
\label{lem:table-permutation}
Let $a$ and $b$ be any positive integers.
Given an $a \times b$ table in which each integer $c \in [a]$ appears exactly $b$~times, it is always possible to permute the integers in each row of the table so that in the resulting table, each column forms a permutation of $[a]$.
\end{lemma}

The proof of \Cref{lem:table-permutation} can be found in Appendix~\ref{app:omitted-positional-prop}.
We now proceed to prove \Cref{thm:positional-prop}.

\begin{proof}[Proof of \Cref{thm:positional-prop}]
We shall construct the output rankings by filling an $m \times k$ table, with each of the $k$ columns corresponding to an output ranking. 
For any $c \in C$ and $\ell \in [m]$, let $t_{c, \ell}$ be the largest integer such that candidate $c$ appears in the top $\ell$ ranks in at least $t_{c, \ell}\cdot n / k$ input rankings. 
It suffices to ensure that $c$ appears in the top $\ell$ ranks in at least $t_{c, \ell}$ output rankings. 
For convenience, we additionally let $t_{c, 0} = 0$ for all $c \in C$, so that we may define $d_{c, \ell} = t_{c, \ell} - t_{c, \ell - 1}$ for all $c \in C$ and $\ell \in [m]$. It is clear that $d_{c, \ell} \geq 0$ for all $c \in C$ and $\ell \in [m]$.

We will fill the table row by row from top to bottom, each row from left to right. 
Iterating over $j \in[m]$ in increasing order, and then over $c \in C$ in arbitrary order, fill $d_{c, j}$ cells with $c$.
After performing this process up to (and including) the step where $j = \ell$ for some $\ell\in[m]$, the table contains $\sum_{j \in [\ell]} d_{c, j} = t_{c, \ell} - t_{c, 0} = t_{c, \ell}$ occurrences of $c$. 
By definition of $t_{c, \ell}$, we have $\sum_{c \in C} t_{c, \ell}\cdot n/k \leq \ell n$, since the top $\ell$ ranks across the $n$ input rankings consist of $\ell n$ positions in total. 
Hence, it holds that $\sum_{c \in C} t_{c, \ell} \leq \ell k$.
Since we fill the table from top to bottom, all elements that we have filled up to this step are contained within the first $\ell$ rows of the table.
In particular, for any $c \in C$ and $\ell \in [m]$, candidate~$c$ appears at least $t_{c, \ell}$ times in the top $\ell$ rows of the table.
Moreover, since $t_{c, m} = k$ for every $c\in C$, each candidate~$c$ appears exactly $k$ times in the table.

We now apply \Cref{lem:table-permutation} to permute the elements in each row of the table so that each column forms a permutation of~$[m]$. 
Since we only permute elements within rows, for any $c \in C$ and $\ell \in [m]$, it still holds that $c$ appears at least $t_{c, \ell}$ times in the top $\ell$ rows of the table. 
Therefore, by letting the $k$ columns of the table be the $k$ output rankings, we obtain an output that satisfies positional proportionality.

From this existence proof, we can also derive a polynomial-time algorithm to produce the output. 
To this end, let $r_{c, \ell}$ be the number of input rankings in which candidate~$c$ appears in the top $\ell$ ranks. 
First, we compute $r_{c, \ell}$ for all $c \in C$ and $\ell \in [m]$ in $O(nm + m^2)$ time by iterating through the input rankings appropriately.
From the values of $r_{c, \ell}$, we compute $t_{c, \ell} = \left \lfloor r_{c, \ell} \cdot k/n \right\rfloor$, and then perform our previously described process of filling the $m \times k$ table.

With the $m \times k$ table, we compute $s_{c, \ell}$, the number of times $c$ appears in row $\ell$ of the table, for all $c \in C$ and $\ell \in [m]$.
Following our inductive proof of \Cref{lem:table-permutation}, we can construct each of the $k$ columns by constructing the described bipartite graph and running the \emph{Hopcroft--Karp algorithm}~\citep{hopcroft1973n} to find a perfect matching. 
After constructing each column, we update the counts $s_{c, \ell}$ to construct the subsequent graphs. 
Since each bipartite graph has $m$~vertices on each of the two sides, it takes $O(m^2)$ time to construct the graph by inspecting the values of $s_{c, \ell}$. 
Each run of the Hopcroft--Karp algorithm takes $O(m^{2.5})$ time~\citep{hopcroft1973n}. 
As we construct $k$ bipartite graphs, the overall time complexity is $O(nm + m^2k + m^{2.5}k) = O(nm + m^{2.5}k)$.
\end{proof}

\section{Proportionality for solid coalitions}
\label{sec:PSC}
In this section, we consider proportionality for solid coalitions (PSC), which captures a different form of aligned preferences than positional proportionality by focusing on groups of voters with shared top-ranked candidates. We show that, as with positional proportionality, an output satisfying PSC always exists and can be computed in polynomial time.

\begin{figure}[t]
\centering
\newsavebox{\dagboxtwo}
\sbox{\dagboxtwo}{%
\begin{tikzpicture}[
    vertex/.style={rectangle, rounded corners=1.5mm, draw=black, minimum height=0.55cm, minimum width=0.7cm,
                   inner sep=2pt, font=\footnotesize},
    fullvertex/.style={rectangle, rounded corners=1.5mm, draw=black, minimum height=0.65cm, minimum width=0.8cm,
                       inner sep=2pt, font=\scriptsize},
    normaledge/.style={->, >=stealth, draw=black!25, line width=0.6pt},
    patha/.style={blue!60!black, line width=1.5pt, ->, >=stealth},
    pathb/.style={green!60!black, line width=1.5pt, dashed, ->, >=stealth},
    bound/.style={label={[font=\tiny, text=black!80, label distance=-2pt]below:$\ge#1$}}
]

\node[vertex, bound=3] (empty) at (0.0, 0) {$\emptyset$};

\node[vertex, bound=1] (one) at (1.6, 1.5) {$1$};
\node[vertex] (two) at (1.6, 0.5) {$2$};
\node[vertex, bound=1] (three) at (1.6, -0.5) {$3$};
\node[vertex, bound=1] (four) at (1.6, -1.5) {$4$};

\node[vertex, bound=1] (onetwo) at (3.2, 2.5) {$1,2$};
\node[vertex] (onethree) at (3.2, 1.5) {$1,3$};
\node[vertex] (onefour) at (3.2, 0.5) {$1,4$};
\node[vertex] (twothree) at (3.2, -0.5) {$2,3$};
\node[vertex] (twofour) at (3.2, -1.5) {$2,4$};
\node[vertex, bound=1] (threefour) at (3.2, -2.5) {$3,4$};

\node[vertex, bound=1] (onetwothree) at (4.800000000000001, 1.5) {$1,2,3$};
\node[vertex, bound=1] (onetwofour) at (4.800000000000001, 0.5) {$1,2,4$};
\node[vertex, bound=1] (onethreefour) at (4.800000000000001, -0.5) {$1,3,4$};
\node[vertex] (twothreefour) at (4.800000000000001, -1.5) {$2,3,4$};

\node[vertex, bound=3] (full) at (6.4, 0) {$1,2,3,4$};

\draw[normaledge] (empty) -- (one);
\draw[normaledge] (empty) -- (two);
\draw[normaledge] (empty) -- (three);
\draw[normaledge] (empty) -- (four);

\draw[normaledge] (one) -- (onetwo);
\draw[normaledge] (one) -- (onethree);
\draw[normaledge] (one) -- (onefour);
\draw[normaledge] (two) -- (onetwo);
\draw[normaledge] (two) -- (twothree);
\draw[normaledge] (two) -- (twofour);
\draw[normaledge] (three) -- (onethree);
\draw[normaledge] (three) -- (twothree);
\draw[normaledge] (three) -- (threefour);
\draw[normaledge] (four) -- (onefour);
\draw[normaledge] (four) -- (twofour);
\draw[normaledge] (four) -- (threefour);

\draw[normaledge] (onetwo) -- (onetwothree);
\draw[normaledge] (onetwo) -- (onetwofour);
\draw[normaledge] (onethree) -- (onetwothree);
\draw[normaledge] (onethree) -- (onethreefour);
\draw[normaledge] (onefour) -- (onetwofour);
\draw[normaledge] (onefour) -- (onethreefour);
\draw[normaledge] (twothree) -- (onetwothree);
\draw[normaledge] (twothree) -- (twothreefour);
\draw[normaledge] (twofour) -- (onetwofour);
\draw[normaledge] (twofour) -- (twothreefour);
\draw[normaledge] (threefour) -- (onethreefour);
\draw[normaledge] (threefour) -- (twothreefour);

\draw[normaledge] (onetwothree) -- (full);
\draw[normaledge] (onetwofour) -- (full);
\draw[normaledge] (onethreefour) -- (full);
\draw[normaledge] (twothreefour) -- (full);

\draw[patha] (empty) -- (one);
\draw[patha] (one) -- (onetwo);
\draw[patha] (onetwo) -- (onetwofour);
\draw[patha] (onetwofour) -- (full);

\draw[pathb] (empty) -- (three);
\draw[pathb] (three) -- (threefour);
\draw[pathb] (threefour) -- (onethreefour);
\draw[pathb] (onethreefour) -- (full);
\end{tikzpicture}
}
\begin{subfigure}[b]{0.35\textwidth}
\centering
\begin{minipage}[b][\ht\dagboxtwo][c]{\textwidth}
\centering
\renewcommand{\arraystretch}{1.15}
\begin{tabular}{c}
\toprule
\textbf{Voters' rankings} \\
\midrule
$\pi_1 = (1,2,4,3)$ \\
$\pi_2 = (1,2,3,4)$ \\
$\pi_3 = (4,3,1,2)$ \\
$\pi_4 = (3,4,1,2)$ \\
$\pi_5 = (3,2,1,4)$ \\
$\pi_6 = (4,2,1,3)$ \\
\bottomrule
\end{tabular}
\end{minipage}
\caption{Instance with $(n, m, k) = (6, 4, 3)$}
\label{fig:profile_table}
\end{subfigure}%
\hfill
\begin{subfigure}[b]{0.65\textwidth}
\centering
\usebox{\dagboxtwo}
\caption{Constructed DAG with PSC constraints}
\label{fig:dag}
\end{subfigure}
\caption{An example instance with $(n,m,k)=(6,4,3)$ and the 
corresponding DAG. The constraints below each vertex indicate the lower bound $q_S$ whenever it is nonzero.
The solid blue path 
$\emptyset\!\to\!\{1\}\!\to\!\{1,2\}\!\to\!\{1,2,4\}\!\to\!\{1,2,3,4\}$ 
and the dashed green path 
$\emptyset\!\to\!\{3\}\!\to\!\{3,4\}\!\to\!\{1,3,4\}\!\to\!\{1,2,3,4\}$ 
correspond to the rankings $(1,2,4,3)$ and $(3,4,1,2)$, respectively.}
\label{fig:dag_combined_two}
\end{figure}
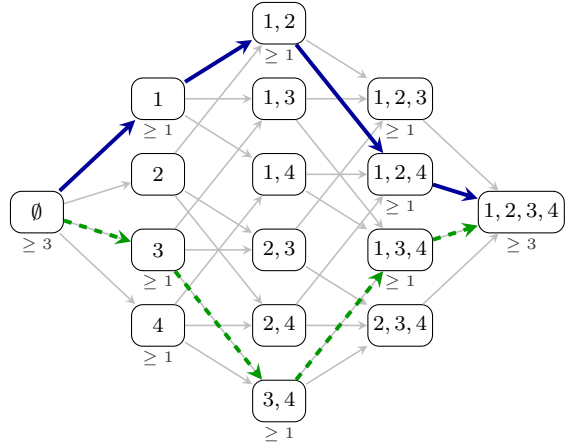

Before we describe the result, we introduce some necessary notation. 
For any ranking $\pi$ over $C$ and any size $\ell \in [m]_0$, let $\Top{\ell}{\pi}$ denote the set of the top $\ell$ candidates in $\pi$. 
For a subset $S \subseteq C$, we define $n_S := |\{i \in N \mid \Top{|S|}{\pi_i} = S \}|$ as the number of voters who rank the candidates from $S$ in their top $|S|$ positions. 
The PSC condition requires that for any $S \subseteq C$ and $t \in [k]$, if $n_S \ge t \cdot n/k$, then at least $t$ of the $k$ output rankings have $S$ as their top $|S|$ candidates. Equivalently, for any $S \subseteq C$, at least $q_S := \floor{n_S \cdot k/n}$ of the output rankings $\pi$ must satisfy $\Top{|S|}{\pi} = S$.

Since PSC depends only on prefixes of rankings, it is helpful to model rankings as chains of growing subsets.
To this end, we construct a directed acyclic graph (DAG) $G^C = (V^C, E^C)$ to model the rankings over $C$. The vertex set $V^C = \mathcal{P}(C)$ consists of all subsets of $C$. 
A directed edge $(S, S') \in E^C$ exists if and only if $S \subset S'$ and $|S'| = |S| + 1$. 
There is a natural bijection between the directed paths from $\emptyset$ to $C$ in $G^C$ and the strict rankings over $C$. Specifically, a ranking $\pi = (c_1, c_2, \ldots, c_m)$ corresponds to the path $\emptyset = \Top{0}{\pi} \to \Top{1}{\pi} \to \cdots \to \Top{m}{\pi} = C$. 
In this graph formulation, selecting $k$ rankings satisfying PSC is equivalent to selecting $k$ paths from $\emptyset$ to $C$ such that for every $S \subseteq C$, at least $q_S$ paths visit the vertex $S$. 
\Cref{fig:dag_combined_two} provides a visual representation of this structure.

It is worth noting that the existence of a valid output does not follow directly from simple greedy heuristics.
To see this, consider the instance in~\Cref{fig:profile_table} with $(n,m,k) = (6,4,3)$.  
A variety of greedy approaches might first produce the following two paths, as illustrated in \Cref{fig:dag}:
\[
    \{1\} \to \{1,2\} \to \{1,2,4\} \to \{1,2,3,4\}
    \quad \text{and} \quad
    \{3\} \to \{3,4\} \to \{1,3,4\} \to \{1,2,3,4\}. 
\]
These routing decisions might appear ``locally optimal'' for this instance because, at every step, each path opportunistically extends to an adjacent vertex with an unmet visit quota. 
This choice is seemingly justified by the assumption that the vertex must eventually be reached by \emph{some} path. 
Ultimately, however, these two locally optimal paths cannot be extended to a PSC output for $k=3$, as the outstanding quotas for vertices $\{4\}$ and $\{1, 2, 3\}$ cannot be satisfied by a single path.

This example shows that purely local strategies are insufficient: constructing a valid path cover requires a more global view of the DAG. 
In the following theorem, we demonstrate how to obtain such a solution by leveraging classical results from network flow.

\begin{theorem}
\label{thm:PSC}
For any instance, there exists an output that satisfies PSC.
Moreover, such an output can be computed in polynomial time.
\end{theorem}

\begin{proof}
We model the problem with a DAG as described above.
Finding a collection of paths that satisfy vertex-visit lower bounds is a classic network flow setup. 
One can imagine routing $k$ units of flow from $\emptyset$ to $C$, where each vertex $S$ demands that at least $q_S$ units of flow pass through it. 
Formally, let $G = (V, E)$ be a directed graph with a \emph{source} $s \in V$ (having no incoming edges) and a \emph{sink} $t \in V$ (having no outgoing edges). Let $L: V \to \mathbb{R}_{\ge 0}$ be a vertex lower bound function. 
A flow in this network corresponds to an assignment of nonnegative weights to the edges, $f: E \to \mathbb{R}_{\ge 0}$. For any vertex $v \in V$, let $\mathrm{In}(v)$ and $\mathrm{Out}(v)$ denote the sets of its in-neighbors and out-neighbors, respectively. 
We define the total incoming and outgoing flows at $v$ as $f_{\mathrm{in}}(v) = \sum_{u \in \mathrm{In}(v)} f(u, v)$ and $f_{\mathrm{out}}(v) = \sum_{w \in \mathrm{Out}(v)} f(v, w)$. The total flow passing through a vertex $v$ is given by $f(v) := f_{\mathrm{in}}(v)$ for all $v \in V \setminus \{s\}$, and $f(s) := f_{\mathrm{out}}(s)$. A flow $f$ is considered \emph{valid} if it satisfies two conditions:
\begin{itemize}[leftmargin=*,nosep]
    \item \emph{Flow conservation}: For all intermediate vertices $v \in V \setminus \{s, t\}$, $f_{\mathrm{in}}(v) = f_{\mathrm{out}}(v)$.
    \item \emph{Demand constraints}: For all $v \in V$, $f(v) \ge L(v)$.
\end{itemize}
The value of the flow, denoted $\mathrm{val}(f)$, is the net amount of flow leaving the source, $f_{\mathrm{out}}(s)$ (which is equivalently the flow entering the sink, $f_{\mathrm{in}}(t)$).

Standard network flow integrality theorems typically apply to edge upper bounds, but they extend to vertex lower bounds via a known reduction~\citep{ahuja1993network}. For completeness, we state this as a lemma, with the proof deferred to Appendix~\ref{app:omitted-PSC}.

\begin{lemma}\label{lem:flow-integrality}
Let $G = (V, E)$ be a directed graph with source $s$ and sink $t$. Let $L: V \to \mathbb{Z}_{\ge 0}$ be an integral vertex lower bound function. 
If there exists a valid (possibly fractional) flow $f$ satisfying the lower bounds, then there exists a valid integral flow $f^*$ satisfying the same bounds $f^*(v) \ge L(v)$ for all $v \in V$, such that $\mathrm{val}(f^*) \le \ceil{\mathrm{val}(f)}$. 
Furthermore, $f^*$ can be computed in time $O((|E|  + |V|) \sum_{v \in V} L(v))$.
\end{lemma}

We apply this framework to $G^C$ with source $s = \emptyset$, sink $t = C$, and integral lower bounds $L(S) = q_S$ for all $S \subseteq C$. We first construct a fractional flow $f$ of value $k$ satisfying these constraints by routing $k/n$ units of flow along each of the $n$ paths corresponding to the input rankings $\pi_1, \ldots, \pi_n$. Formally, for each edge $(S, S') \in E^C$, we set
\[
    f(S, S') = \frac{k}{n} \cdot \left|\left\{i \in N \,\middle|\, \Top{|S|}{\pi_i} = S \text{ and } \Top{|S'|}{\pi_i} = S'\right\}\right|.
\]
This flow satisfies conservation and has value $k$. For any intermediate vertex $S \in \mathcal{P}(C) \setminus \{ \emptyset, C\}$, the incoming and outgoing flows correspond exactly to $k$ times the fraction of voters who rank $S$ in their top $|S|$ positions. 
Thus, $f_{\mathrm{in}}(S) = f_{\mathrm{out}}(S) = n_S \cdot k/n$, which means that $f(S) = n_S \cdot k/n$. 
For the source and sink, we similarly have $f(\emptyset) = f(C) = n \cdot k/n = k$. The demand constraints are satisfied everywhere because for each $S \subseteq C$,
\[
    f(S) = n_S \cdot \frac{k}{n}  \ge \floor{n_S \cdot \frac{k}{n}} = L(S).
\]

Now, by \Cref{lem:flow-integrality}, there exists an integral flow $f^*$ of value $\mathrm{val}(f^*) \le \ceil{\mathrm{val}(f)} = k$ satisfying $f^*(S) \ge L(S)$ for all $S \subseteq C$. Note that any valid flow must have a value of exactly $k$, since the lower bound constraint on the source requires $\mathrm{val}(f^*) = f^*(\emptyset) \ge L(\emptyset) = \floor{n \cdot k/n} = k$. 

Because $G^C$ is a DAG and $f^*$ is an integral flow of value $k$, the flow~$f^*$ can be decomposed into $k$ directed paths $P_1, \ldots, P_k$ from $\emptyset$ to $C$.
Indeed, this can be done greedily by repeatedly extracting a path from $\emptyset$ to $C$ such that every edge on the path has a positive flow (which must be at least $1$), and subtracting $1$ from each edge on the path.
Since $f^*(S) \ge L(S)$ for all $S$, at least $L(S) = q_S$ paths visit each vertex $S$. 
By interpreting these $k$ paths as rankings, we satisfy the PSC condition.

Finally, we analyze the running time.
Note that we do not need to construct the entire DAG $G^C$ (which has $2^m$ vertices). 
Instead, we can restrict the network to the subgraph induced by the vertices and edges visited by the paths corresponding to the $n$ input rankings. 
Since each voter's ranking defines a path of $m + 1$ vertices and $m$ edges, this subgraph has at most $n(m+1)$ vertices and $nm$ edges. Furthermore, for any fixed subset size $\ell \in [m]_0$, each voter's path visits exactly one subset $S$ of size $\ell$.
Therefore, $\sum_{S : |S| = \ell} n_S = n$, which bounds the total demand as $\sum_{S : |S| =\ell} L(S) \le \sum_{S : |S| = \ell} n_S \cdot k/n = k$. Summing over all $m+1$ subset sizes yields $\sum_{S \subseteq C} L(S) \le k(m+1) = O(mk)$. 

Applying \Cref{lem:flow-integrality}, we can compute the integral flow $f^*$ on this restricted subgraph in $O((|E| + |V|) \sum_{S} L(S)) = O((nm) \cdot mk) = O(nm^2 k)$ time. 
Extracting the $k$ paths takes $O(nmk)$ time.
Thus, the overall process takes $O(nm^2k)$ time.
\end{proof}

\section{Pairwise proportionality}
\label{sec:pairwise-prop}

In this section, we turn our attention to pairwise proportionality.
While one might intuitively expect that a pairwise proportional output always exists, we prove a surprising impossibility result: for broad ranges of parameters $n$ and $k$, such an output cannot be guaranteed.

\begin{theorem}
\label{thm:pairwise-prop}
For any $k\ge 2$ and $n\ge k+2$, there exists an instance with $n$ voters and target compression size~$k$ such that no output satisfies pairwise proportionality.    
\end{theorem}

We begin by setting up the necessary preliminaries for the case $k = 2$.
For each pair of candidates $c, c' \in C$, let $t_{c, c'}$ be the largest integer such that $c$ is ranked higher than $c'$ in at least $t_{c, c'}\cdot n/2$ input rankings. 
When $t_{c, c'} = 2$, pairwise proportionality imposes the constraint that $c$ must be ranked higher than $c'$ in both output rankings, so we write $c \sto c'$ and call this a \textit{strong constraint}. 
When $t_{c, c'} = 1$, we only need $c$ to be ranked higher than $c'$ in one output ranking, so we write $c \wto c'$ and call this a \textit{weak constraint}. 
We say that a ranking $\pi$ \emph{satisfies a constraint $c \to c'$}, which can be either strong or weak, if $c$ is ranked higher than $c'$ in $\pi$.

We will model the constraints as a directed graph whose vertex set is $C = [m]$. 
We add a solid edge from $c$ to $c'$ whenever there is a strong constraint $c \sto c'$, and add a dotted edge from $c$ to $c'$ whenever there is a weak constraint $c \wto c'$.
For example, in the figure below, there are strong constraints $4\sto 1$ and $2\sto 3$, and weak constraints $1 \wto 2$ and $3 \wto 4$. 
We call the directed graph induced by the constraints a \emph{constraint graph}, and say that a constraint graph is satisfied by two output rankings if the number of rankings satisfying each constraint is at least that required by the type of the constraint.

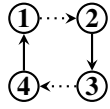
{
\centering
    \begin{tikzpicture}[>=stealth, thick, scale=0.45]
        \tikzstyle{vertex}=[circle, draw, minimum size=11pt, inner sep=0pt, font=\bfseries]

        \node[vertex] (1) at (0, 2) {1};
        \node[vertex] (2) at (2, 2) {2};
        \node[vertex] (3) at (2, 0) {3};
        \node[vertex] (4) at (0, 0) {4};

        \draw[dotted, ->] (1) -- (2); 
        \draw[->] (2) -- (3);         
        \draw[dotted, ->] (3) -- (4); 
        \draw[->] (4) -- (1);         
    \end{tikzpicture}
    \par
}

Observe that the constraint graph generated by any input cannot contain a cycle consisting of only strong constraints. 
Indeed, suppose there is a cycle $c_1 \sto c_2 \sto \cdots \sto c_r \sto c_1$. 
By definition, this means that $c_1$ must be ranked higher than $c_1$ in every input ranking, which is impossible.

The following lemma allows us to characterize constraint graphs satisfiable by two output rankings.

\begin{lemma}
\label{lem:pairwise_prop_k_equals_2_condition}
When $k = 2$, a constraint graph can be satisfied by two output rankings if and only if there exists a ranking $\pi$ satisfying all strong constraints, such that for every directed cycle in the graph, $\pi$ satisfies at least one weak constraint in the cycle.
\end{lemma}

\begin{proof}
($\Rightarrow$) Suppose that a constraint graph is satisfied by the output rankings $\pi$, $\pi'$. 
By definition, $\pi$ must satisfy all strong constraints. Suppose for contradiction that there exists a directed cycle in the graph such that none of its weak constraints is satisfied by $\pi$. 
This means that every weak constraint in the cycle must be satisfied by $\pi'$. 
Since every strong constraint in the cycle is also satisfied by $\pi'$, we have that $\pi'$ satisfies all constraints in the cycle, which is a contradiction.

($\Leftarrow$) Let $\pi$ be a ranking as described in the lemma statement.
We want to construct another ranking~$\pi'$ such that the two output rankings $\pi, \pi'$ together satisfy the constraint graph. To this end, it suffices to have $\pi'$ satisfy all strong constraints in addition to all weak constraints not satisfied by $\pi$. 
We will show that the subgraph induced by this subset of constraints is acyclic, and we can then take $\pi'$ to be any topological ordering of this subgraph. 

Suppose for contradiction that there exists a directed cycle in this subgraph. 
Since this cycle also exists in the original constraint graph, $\pi$ satisfies at least one weak constraint in the cycle, by definition of~$\pi$.
However, this means that this weak constraint should not exist in the subgraph, a contradiction. 
\end{proof}

We now proceed to prove \Cref{thm:pairwise-prop} for the basic case where $(n, k) = (4, 2)$.

\begin{proof}[Proof of \Cref{thm:pairwise-prop} for $(n,k) = (4,2)$]
We will specify a constraint graph, show that it cannot be satisfied by any two output rankings, and give four input rankings that generate the constraint graph.
These together suffice to complete the proof.

For $j \in [3]_0$, let $D_j$ denote the cycle 
\[6j + 1 \wto 6j + 2 \sto 6j + 3 \wto 6j + 4 \sto 6j + 5 \wto 6j + 6 \sto 6j + 1. \]
In particular, there are $m = 24$ candidates.
See \Cref{fig:four-six-vertex-cycles} for an illustration.

\begin{figure}[t]
    \centering
    \begin{tikzpicture}[>=stealth, thick, scale = 0.7]
    
    \tikzset{vertex/.style={circle, draw, minimum size=14pt, inner sep=0pt, font=\small\bfseries}}
    
    \newcommand{\hexcycle}[2]{%
      \foreach \i in {0,...,5} {%
        \pgfmathsetmacro{\angle}{90 - 60*\i} 
        \pgfmathsetmacro{\xcoord}{#1 + 1.5*cos(\angle)}
        \pgfmathsetmacro{\ycoord}{1.5*sin(\angle)}
        \pgfmathtruncatemacro{\label}{#2 + \i}
        \node[vertex] (v\label) at (\xcoord,\ycoord) {\label};
      }%
      \foreach \i in {0,...,5} {%
        \pgfmathtruncatemacro{\lab}{#2 + \i}
        \pgfmathtruncatemacro{\j}{mod(\i+1,6)}
        \pgfmathtruncatemacro{\labj}{#2 + \j}
        \ifodd\i
          \draw[->] (v\lab) -- (v\labj);
        \else
          \draw[dotted,->] (v\lab) -- (v\labj);
        \fi
      }%
    }
    
    \hexcycle{0}{1}    
    \hexcycle{4.5}{7}  
    \hexcycle{9}{13}   
    \hexcycle{13.5}{19} 

    \node at (0,0) {$D_0$};
    \node at (4.5,0) {$D_1$};
    \node at (9,0) {$D_2$};
    \node at (13.5,0) {$D_3$};

    \end{tikzpicture}

    \caption{A graph representing part of the constraints in the proof of \Cref{thm:pairwise-prop}.}
    \label{fig:four-six-vertex-cycles}
\end{figure}
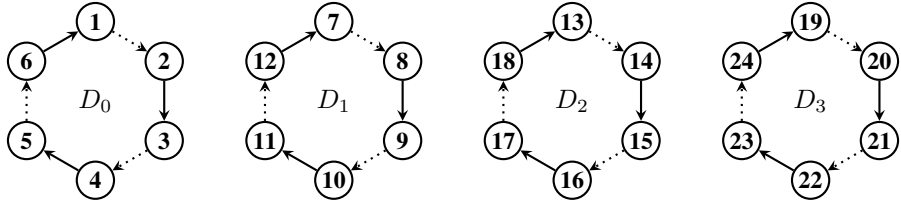

For $j\in[3]_0$, denote the sets of odd-numbered and even-numbered vertices in $D_j$ by $O_j$ and $E_j$, respectively. 
We will take the indices $j$ of $D_j, O_j, E_j$ modulo $4$ hereafter. 
For each $j\in[3]_0$, we additionally impose a strong constraint $u \sto v$ for every pair $(u,v) \in (E_j, O_{j+1})$.

We claim that for every ranking $\pi$ that satisfies all strong constraints, there exists $j\in [3]_0$ such that none of the weak constraints in $D_j$ is satisfied by~$\pi$. 
Suppose for contradiction that there exists a ranking $\pi$ such that for each $j\in [3]_0$, there exists a weak constraint $u_j \wto v_j$ in $D_j$ that is satisfied by~$\pi$.
By our construction, $u_j \in O_j$ and $v_j \in E_j$. 
This means that there is a strong constraint $v_j \sto u_{j + 1}$ for each $j\in[3]_0$ as well. 
Thus, $\pi$ satisfies all constraints in the cycle
\[ u_0 \wto v_0 \sto u_1 \wto v_1 \sto u_2 \wto v_2 \sto u_3 \wto v_3 \sto u_0,\]
a contradiction. 
It follows that for every ranking $\pi$, there exists a cycle in the constraint graph in which no weak constraint is satisfied by $\pi$.
By \Cref{lem:pairwise_prop_k_equals_2_condition}, no two output rankings can satisfy this constraint graph.

It remains to construct four input rankings that generate the constraint graph. 
For each $j\in[3]_0$ and $j' \in [6]$, let $S_{j, j'}$ be the sequence of vertices that we obtain by traversing the cycle $D_j$ starting from vertex $6j + j'$ and following its edges. 
For example, $S_{2, 3}$ is the sequence $15, 16, 17, 18, 13, 14$.
We define $\pi_j = (E_j, S_{j + 1, 6}, S_{j + 2, 4}, S_{j + 3, 2}, O_j)$ for $j \in[3]_0$, and take $\pi_0, \pi_1, \pi_2, \pi_3$ as the four input rankings.

For each $\pi_j$, we first verify that the strong constraints from $E_{j'}$ to $O_{j' + 1}$ for all $j'\in[3]_0$ are satisfied.
For $j' = j$, $E_j$ is ranked higher than $S_{j + 1, 6}$; 
for $j' = j + 1$, $S_{j + 1, 6}$ is ranked higher than $S_{j + 2, 4}$; 
for $j' = j + 2$, $S_{j + 2, 4}$ is ranked higher than $S_{j + 3, 2}$;
for $j' = j + 3$, $S_{j + 3, 2}$ is ranked higher than~$O_j$.
Additionally, for all $j' \in [3]_0$ and $j'' \in [3]$, all strong constraints in $D_{j'}$ are satisfied by $S_{j', 2j''}$. Therefore, all strong constraints in $D_{j + 1}, D_{j + 2}, D_{j + 3}$ are satisfied by $\pi_j$. Moreover, $E_j$ is ranked higher than $O_j$ in $\pi_j$, so all strong constraints in $D_j$ are satisfied as well. 
This means that each $\pi_j$ satisfies all strong constraints.

Finally, observe that for any $j\in[3]_0$, each weak constraint in $D_j$ is satisfied by two of $S_{j, 2}, S_{j, 4}, S_{j, 6}$, which appear in $\pi_{j - 3}, \pi_{j - 2}, \pi_{j - 1}$, respectively.
Therefore, the four input rankings together generate all constraints in the graph.
This completes the proof.
\end{proof}

Our instance in the proof of \Cref{thm:pairwise-prop} requires $m = 24$, which is rather large. Interestingly, an exhaustive search revealed that pairwise proportionality can always be satisfied for $(n, k) = (4, 2)$ up to $m = 7$.
This suggests that even though a pairwise proportional outcome is not guaranteed to exist, such an outcome is likely to exist, at least when the parameters are sufficiently small.

The proof for general $k\ge 2$ and $n\ge k+2$ uses a similar approach as that of the case $(n,k) = (4,2)$.
The full details can be found in Appendix~\ref{app:proof-pairwise-prop}.
We also remark that these bounds on the parameters are tight.
Indeed, if $k = 1$ or $n\in\{k,k+1\}$, we claim that taking any $k$ input rankings as the output already works.
This is clear if $k = 1$ or $n = k$.
If $n = k+1$, then whenever there is a constraint that $i$ must be ranked higher than $j$ in at least $t$ output rankings, it must be that $i$ is ranked higher than $j$ in at least $\lceil t \cdot n/k\rceil = \lceil t(k+1)/k\rceil \ge t+1$ input rankings.
Since we exclude only one input ranking from the output, the constraint is satisfied.

Despite this negative result, we show that pairwise proportionality is always achievable when the input rankings are well-structured. Specifically, this guarantee holds for the two most widely studied classes of structured preferences, \emph{single-peaked} and \emph{single-crossing} \citep{ElkindLaPe17}; the definitions of these classes are deferred to Appendix~\ref{app:proof-restricted-preferences}.
For single-peaked instances, not only can we find an output satisfying pairwise proportionality, but we can also ensure that the output simultaneously fulfills the other two proportionality notions and is single-peaked with respect to the same candidate ordering.
For single-crossing instances, the output can be taken from the input rankings and additionally satisfies PSC.

\begin{theorem}
\label{thm:single-peaked-simul-prop}
For any single-peaked instance, there exists an output that satisfies pairwise proportionality, positional proportionality, and PSC simultaneously, such that every output ranking is single-peaked with respect to the same ordering of candidates as the input.
Moreover, such an output can be computed in polynomial time.
\end{theorem}

\begin{theorem}
\label{thm:pairwise-prop-single-crossing}
For any single-crossing instance, there exists an output that satisfies pairwise proportionality and PSC simultaneously, such that every output ranking appears in the input.
Moreover, such an output can be computed in polynomial time.
\end{theorem}

The proofs of \Cref{thm:single-peaked-simul-prop,thm:pairwise-prop-single-crossing} can be found in Appendix~\ref{app:proof-restricted-preferences}.
While one may hope to add positional proportionality to \Cref{thm:pairwise-prop-single-crossing}, we demonstrate in Appendix~\ref{app:single-crossing-output-from-input} that this is in fact impossible, at least without giving up the output-from-input condition.
Moreover, since the instance used in our proof of \Cref{prop:output-from-input} is single-peaked (see details in Appendix~\ref{app:output-from-input}), adding this condition to \Cref{thm:single-peaked-simul-prop} is also infeasible.

\section{Future directions}
\label{sec:future-directions}

While our paper establishes a solid foundation for understanding proportionality notions in ranking compression, several interesting directions remain to be explored.
In particular, we highlight some notable avenues for future work.

\textbf{Axiom relaxations.} Since pairwise proportionality is not always satisfiable (\Cref{thm:pairwise-prop}), can a relaxation of it nevertheless be guaranteed?
If we relax the notion multiplicatively (i.e., by requiring candidate~$c$ to be ranked higher than candidate~$c'$ in at least $r\cdot t \cdot n/k$ \emph{input} rankings), an approximation of $r = 2$ can be achieved.
Indeed, we can simply let the output consist of an arbitrary ranking $\ceil{k/2}$ times and the reverse of that ranking $\floor{k/2}$ times.
On the other hand, for an additive relaxation (i.e., we require $c$ to be ranked higher than $c'$ in at least $t-r$ \emph{output} rankings), it remains unclear whether an approximation of $r = 1$ is always feasible.

\textbf{Probabilistic guarantees.} Given the adversarial nature of our counterexamples for pairwise proportionality (\Cref{thm:pairwise-prop}), is it true that when the input is drawn from certain probabilistic distributions, a pairwise proportional output exists with high probability?

\textbf{Practical AI alignment.} Zooming out, a major challenge lies in translating our theoretical framework into real-world alignment pipelines.
For example, can we maintain similar proportionality guarantees if users only provide partial rankings, and what structural assumptions would be required for such data? 
Finally, deploying this framework entails mapping the output rankings back to tractable reward functions. To ensure that these rewards can be learned in practice, it may be necessary to draw the output rankings from a restricted class that mathematically corresponds to viable functions, similar to the guarantees that we established in the single-peaked setting (\Cref{thm:single-peaked-simul-prop}).

\section*{Acknowledgments}

This work was partially supported by the Ministry of Education, Singapore, under the Academic Research Fund Tier 1 (FY 2026) grant number 251RES2604, and by an NUS Start-up Grant.
We thank the anonymous reviewers for their valuable comments.

\bibliographystyle{plainnat}
\bibliography{abb,main}

\begin{thebibliography}{48}
\providecommand{\natexlab}[1]{#1}
\providecommand{\url}[1]{\texttt{#1}}
\expandafter\ifx\csname urlstyle\endcsname\relax
  \providecommand{\doi}[1]{doi: #1}\else
  \providecommand{\doi}{doi: \begingroup \urlstyle{rm}\Url}\fi

\bibitem[Ahuja et~al.(1993)Ahuja, Magnanti, and Orlin]{ahuja1993network}
Ravindra~K. Ahuja, Thomas~L. Magnanti, and James~B. Orlin.
\newblock \emph{Network Flows: Theory, Algorithms and Applications}.
\newblock Prentice Hall, 1993.

\bibitem[Anwar et~al.(2024)Anwar, Saparov, Rando, Paleka, Turpin, Hase, Lubana, Jenner, Casper, Sourbut, et~al.]{anwar2024foundational}
Usman Anwar, Abulhair Saparov, Javier Rando, Daniel Paleka, Miles Turpin, Peter Hase, Ekdeep~Singh Lubana, Erik Jenner, Stephen Casper, Oliver Sourbut, et~al.
\newblock Foundational challenges in assuring alignment and safety of large language models.
\newblock \emph{Transactions on Machine Learning Research}, 2024.

\bibitem[Arrow(1950)]{Arrow50}
Kenneth~J. Arrow.
\newblock A difficulty in the concept of social welfare.
\newblock \emph{Journal of Political Economy}, 58\penalty0 (4):\penalty0 328--346, 1950.

\bibitem[Arrow et~al.(2002)Arrow, Sen, and Suzumura]{ArrowSeSu02}
Kenneth~J. Arrow, Amartya~K. Sen, and Kotaro Suzumura, editors.
\newblock \emph{Handbook of Social Choice and Welfare}, volume~1.
\newblock North-Holland, 2002.

\bibitem[Aziz and Lee(2020)]{AzizLe20}
Haris Aziz and Barton~E. Lee.
\newblock The expanding approvals rule: Improving proportional representation and monotonicity.
\newblock \emph{Social Choice and Welfare}, 54\penalty0 (1):\penalty0 1--45, 2020.

\bibitem[Aziz et~al.(2024)Aziz, Lee, Chu, and Vollen]{aziz2024proportionally}
Haris Aziz, Barton~E. Lee, Sean~Morota Chu, and Jeremy Vollen.
\newblock Proportionally representative clustering.
\newblock In \emph{Proceedings of the 20th Conference on Web and Internet Economics (WINE)}, pages 155--171, 2024.

\bibitem[Brandt et~al.(2016)Brandt, Conitzer, Endriss, Lang, and Procaccia]{BrandtCoEn16}
Felix Brandt, Vincent Conitzer, Ulle Endriss, J\'{e}r\^{o}me Lang, and Ariel~D. Procaccia, editors.
\newblock \emph{Handbook of Computational Social Choice}.
\newblock Cambridge University Press, 2016.

\bibitem[Brill and Peters(2023)]{BrillPe23}
Markus Brill and Jannik Peters.
\newblock Robust and verifiable proportionality axioms for multiwinner voting.
\newblock In \emph{Proceedings of the 24th ACM Conference on Economics and Computation (EC)}, page 301, 2023.

\bibitem[Chakraborty et~al.(2024)Chakraborty, Qiu, Yuan, Koppel, Manocha, Huang, Bedi, and Wang]{chakraborty2024maxmin}
Souradip Chakraborty, Jiahao Qiu, Hui Yuan, Alec Koppel, Dinesh Manocha, Furong Huang, Amrit~Singh Bedi, and Mengdi Wang.
\newblock Max{M}in-{RLHF}: Alignment with diverse human preferences.
\newblock In \emph{Proceedings of the 41st International Conference on Machine Learning (ICML)}, pages 6116--6135, 2024.

\bibitem[Chen et~al.(2025)Chen, Chen, Rege, Wang, and Vinayak]{chen2025pal}
Daiwei Chen, Yi~Chen, Aniket Rege, Zhi Wang, and Ramya~Korlakai Vinayak.
\newblock {PAL}: Sample-efficient personalized reward modeling for pluralistic alignment.
\newblock In \emph{Proceedings of the 13th International Conference on Learning Representations (ICLR)}, 2025.

\bibitem[Chen et~al.(2019)Chen, Fain, Lyu, and Munagala]{chen2019proportionally}
Xingyu Chen, Brandon Fain, Liang Lyu, and Kamesh Munagala.
\newblock Proportionally fair clustering.
\newblock In \emph{Proceedings of the 36th International Conference on Machine Learning (ICML)}, pages 1032--1041, 2019.

\bibitem[Chevaleyre et~al.(2009)Chevaleyre, Lang, Maudet, and Ravilly-Abadie]{maudet2009compiling}
Yann Chevaleyre, J\'{e}r\^{o}me Lang, Nicolas Maudet, and Guillaume Ravilly-Abadie.
\newblock Compiling the votes of a subelectorate.
\newblock In \emph{Proceedings of the 21st International Joint Conference on Artificial Intelligence (IJCAI)}, pages 97--102, 2009.

\bibitem[Chevaleyre et~al.(2011)Chevaleyre, Lang, Maudet, and Monnot]{Chevaleyre2}
Yann Chevaleyre, J\'{e}r\^{o}me Lang, Nicolas Maudet, and J\'{e}r\^{o}me Monnot.
\newblock Compilation and communication protocols for voting rules with a dynamic set of candidates.
\newblock In \emph{Proceedings of the 13th Conference on Theoretical Aspects of Rationality and Knowledge (TARK)}, pages 153--160, 2011.

\bibitem[Chhabra et~al.(2021)Chhabra, Masalkovait{\.e}, and Mohapatra]{chhabra2021overview}
Anshuman Chhabra, Karina Masalkovait{\.e}, and Prasant Mohapatra.
\newblock An overview of fairness in clustering.
\newblock \emph{IEEE Access}, 9:\penalty0 130698--130720, 2021.

\bibitem[Cormen et~al.(2022)Cormen, Leiserson, Rivest, and Stein]{CLRS}
Thomas~H. Cormen, Charles~E. Leiserson, Ronald~L. Rivest, and Clifford Stein.
\newblock \emph{Introduction to Algorithms}.
\newblock MIT Press, 4th edition, 2022.

\bibitem[Coste et~al.(2024)Coste, Anwar, Kirk, and Krueger]{coste2024reward}
Thomas Coste, Usman Anwar, Robert Kirk, and David Krueger.
\newblock Reward model ensembles help mitigate overoptimization.
\newblock In \emph{Proceedings of the 12th International Conference on Learning Representations (ICLR)}, 2024.

\bibitem[Dummett(1984)]{Dummett84}
Michael Dummett.
\newblock \emph{Voting Procedures}.
\newblock Oxford University Press, 1984.

\bibitem[Elkind et~al.(2017)Elkind, Lackner, and Peters]{ElkindLaPe17}
Edith Elkind, Martin Lackner, and Dominik Peters.
\newblock Structured preferences.
\newblock In Ulle Endriss, editor, \emph{Trends in Computational Social Choice}, chapter~10, pages 187--207. AI Access, 2017.

\bibitem[Faliszewski et~al.(2023)Faliszewski, Kaczmarczyk, Sornat, Szufa, and W{\k{a}}s]{Faliszewski2023Diversity}
Piotr Faliszewski, Andrzej Kaczmarczyk, Krzysztof Sornat, Stanis{\l}aw Szufa, and Tomasz W{\k{a}}s.
\newblock Diversity, agreement, and polarization in elections.
\newblock In \emph{Proceedings of the 32nd International Joint Conference on Artificial Intelligence (IJCAI)}, pages 2684--2692, 2023.

\bibitem[Faliszewski et~al.(2026)Faliszewski, Sornat, Szufa, and W{\k{a}}s]{faliszewski2026diversity}
Piotr Faliszewski, Krzysztof Sornat, Stanis{\l}aw Szufa, and Tomasz W{\k{a}}s.
\newblock Diversity of structured domains via $k$-{K}emeny scores.
\newblock In \emph{Proceedings of the 40th AAAI Conference on Artificial Intelligence (AAAI)}, pages 16880--16888, 2026.

\bibitem[Ge et~al.(2024)Ge, Halpern, Micha, Procaccia, Shapira, Vorobeychik, and Wu]{ge2024axioms}
Luise Ge, Daniel Halpern, Evi Micha, Ariel~D. Procaccia, Itai Shapira, Yevgeniy Vorobeychik, and Junlin Wu.
\newblock Axioms for {AI} alignment from human feedback.
\newblock In \emph{Proceedings of the 38th Annual Conference on Neural Information Processing Systems (NeurIPS)}, 2024.

\bibitem[Gordon et~al.(2022)Gordon, Lam, Park, Patel, Hancock, Hashimoto, and Bernstein]{gordon2022jury}
Mitchell~L. Gordon, Michelle~S. Lam, Joon~Sung Park, Kayur Patel, Jeffrey~T. Hancock, Tatsunori Hashimoto, and Michael~S. Bernstein.
\newblock Jury learning: Integrating dissenting voices into machine learning models.
\newblock In \emph{Proceedings of the 2022 CHI Conference on Human Factors in Computing Systems (CHI)}, pages 115:1--115:19, 2022.

\bibitem[Halpern et~al.(2025)Halpern, Micha, Procaccia, and Shapira]{HalpernMiPr25}
Daniel Halpern, Evi Micha, Ariel~D. Procaccia, and Itai Shapira.
\newblock Pairwise calibrated rewards for pluralistic alignment.
\newblock In \emph{Proceedings of the 39th Annual Conference on Neural Information Processing Systems (NeurIPS)}, 2025.

\bibitem[Hopcroft and Karp(1973)]{hopcroft1973n}
John~E. Hopcroft and Richard~M. Karp.
\newblock An $n^{5/2}$ algorithm for maximum matchings in bipartite graphs.
\newblock \emph{SIAM Journal on Computing}, 2\penalty0 (4):\penalty0 225--231, 1973.

\bibitem[Kellerhals and Peters(2024)]{kellerhals2024proportional}
Leon Kellerhals and Jannik Peters.
\newblock Proportional fairness in clustering: A social choice perspective.
\newblock In \emph{Proceedings of the 38th Annual Conference on Neural Information Processing Systems (NeurIPS)}, 2024.

\bibitem[Khalifa et~al.(2021)Khalifa, Elsahar, and Dymetman]{khalifa2020distributional}
Muhammad Khalifa, Hady Elsahar, and Marc Dymetman.
\newblock A distributional approach to controlled text generation.
\newblock In \emph{Proceedings of the 9th International Conference on Learning Representations (ICLR)}, 2021.

\bibitem[Kirk et~al.(2024{\natexlab{a}})Kirk, Whitefield, R{\"o}ttger, Bean, Margatina, Ciro, Mosquera, Bartolo, Williams, He, et~al.]{kirk2024prism}
Hannah~Rose Kirk, Alexander Whitefield, Paul R{\"o}ttger, Andrew Bean, Katerina Margatina, Juan Ciro, Rafael Mosquera, Max Bartolo, Adina Williams, He~He, et~al.
\newblock The {PRISM} alignment dataset: What participatory, representative and individualised human feedback reveals about the subjective and multicultural alignment of large language models.
\newblock In \emph{Proceedings of the 38th Annual Conference on Neural Information Processing Systems (NeurIPS)}, 2024{\natexlab{a}}.

\bibitem[Kirk et~al.(2024{\natexlab{b}})Kirk, Mediratta, Nalmpantis, Luketina, Hambro, Grefenstette, and Raileanu]{kirk2024understanding}
Robert Kirk, Ishita Mediratta, Christoforos Nalmpantis, Jelena Luketina, Eric Hambro, Edward Grefenstette, and Roberta Raileanu.
\newblock Understanding the effects of {RLHF} on {LLM} generalisation and diversity.
\newblock In \emph{Proceedings of the 12th International Conference on Learning Representations (ICLR)}, 2024{\natexlab{b}}.

\bibitem[Lake et~al.(2025)Lake, Choi, and Durrett]{lake2025from}
Thom Lake, Eunsol Choi, and Greg Durrett.
\newblock From distributional to {O}verton pluralism: Investigating large language model alignment.
\newblock In \emph{Proceedings of the 2025 Conference of the Nations of the Americas Chapter of the Association for Computational Linguistics: Human Language Technologies (NAACL)}, pages 6794--6814, 2025.

\bibitem[Lederer(2026)]{Lederer26}
Patrick Lederer.
\newblock Proportional representation in rank aggregation.
\newblock In \emph{Proceedings of the 27th ACM Conference on Economics and Computation (EC)}, 2026.

\bibitem[Lederer et~al.(2024)Lederer, Peters, and W{\k{a}}s]{LedererPeWa24}
Patrick Lederer, Dominik Peters, and Tomasz W{\k{a}}s.
\newblock The squared {K}emeny rule for averaging rankings.
\newblock In \emph{Proceedings of the 25th ACM Conference on Economics and Computation (EC)}, page 755, 2024.

\bibitem[Mandal et~al.(2019)Mandal, Procaccia, Shah, and Woodruff]{mandal2019efficient}
Debmalya Mandal, Ariel~D. Procaccia, Nisarg Shah, and David~P. Woodruff.
\newblock Efficient and thrifty voting by any means necessary.
\newblock In \emph{Proceedings of the 33rd Annual Conference on Neural Information Processing Systems (NeurIPS)}, pages 7178--7189, 2019.

\bibitem[Mandal et~al.(2020)Mandal, Shah, and Woodruff]{mandal2020optimal}
Debmalya Mandal, Nisarg Shah, and David~P. Woodruff.
\newblock Optimal communication-distortion tradeoff in voting.
\newblock In \emph{Proceedings of the 21st ACM Conference on Economics and Computation (EC)}, pages 795--813, 2020.

\bibitem[Micha and Shah(2020)]{micha2020proportionally}
Evi Micha and Nisarg Shah.
\newblock Proportionally fair clustering revisited.
\newblock In \emph{Proceedings of the 47th International Colloquium on Automata, Languages, and Programming (ICALP)}, pages 85:1--85:16, 2020.

\bibitem[Ouyang et~al.(2022)Ouyang, Wu, Jiang, Almeida, Wainwright, Mishkin, Zhang, Agarwal, Slama, Ray, et~al.]{ouyang2022training}
Long Ouyang, Jeffrey Wu, Xu~Jiang, Diogo Almeida, Carroll~L. Wainwright, Pamela Mishkin, Chong Zhang, Sandhini Agarwal, Katarina Slama, Alex Ray, et~al.
\newblock Training language models to follow instructions with human feedback.
\newblock In \emph{Proceedings of the 36th Annual Conference on Neural Information Processing Systems (NeurIPS)}, 2022.

\bibitem[Perez et~al.(2022)Perez, Huang, Song, Cai, Ring, Aslanides, Glaese, McAleese, and Irving]{perez2022red}
Ethan Perez, Saffron Huang, Francis Song, Trevor Cai, Roman Ring, John Aslanides, Amelia Glaese, Nat McAleese, and Geoffrey Irving.
\newblock Red teaming language models with language models.
\newblock In \emph{Proceedings of the 2022 Conference on Empirical Methods in Natural Language Processing (EMNLP)}, pages 3419--3448, 2022.

\bibitem[Poddar et~al.(2024)Poddar, Wan, Ivison, Gupta, and Jaques]{poddar2024personalizing}
Sriyash Poddar, Yanming Wan, Hamish Ivison, Abhishek Gupta, and Natasha Jaques.
\newblock Personalizing reinforcement learning from human feedback with variational preference learning.
\newblock In \emph{Proceedings of the 38th Annual Conference on Neural Information Processing Systems (NeurIPS)}, 2024.

\bibitem[Prabhakaran et~al.(2022)Prabhakaran, Qadri, and Hutchinson]{prabhakaran2022cultural}
Vinodkumar Prabhakaran, Rida Qadri, and Ben Hutchinson.
\newblock Cultural incongruencies in artificial intelligence.
\newblock \emph{arXiv preprint arXiv:2211.13069}, 2022.

\bibitem[Procaccia et~al.(2025)Procaccia, Schiffer, Wang, and Zhang]{ProcacciaScWa25}
Ariel~D. Procaccia, Benjamin Schiffer, Serena Wang, and Shirley Zhang.
\newblock Metritocracy: Representative metrics for lite benchmarks.
\newblock In \emph{Proceedings of the 39th Annual Conference on Neural Information Processing Systems (NeurIPS)}, 2025.

\bibitem[Santurkar et~al.(2023)Santurkar, Durmus, Ladhak, Lee, Liang, and Hashimoto]{santurkar2023whose}
Shibani Santurkar, Esin Durmus, Faisal Ladhak, Cinoo Lee, Percy Liang, and Tatsunori Hashimoto.
\newblock Whose opinions do language models reflect?
\newblock In \emph{Proceedings of the 40th International Conference on Machine Learning (ICML)}, pages 29971--30004, 2023.

\bibitem[Shypula et~al.(2025)Shypula, Li, Zhang, Padmakumar, Yin, and Bastani]{shypula2025evaluating}
Alexander Shypula, Shuo Li, Botong Zhang, Vishakh Padmakumar, Kayo Yin, and Osbert Bastani.
\newblock Evaluating the diversity and quality of {LLM} generated content.
\newblock In \emph{Proceedings of the 1st Conference on Language Modeling (COLM)}, 2025.

\bibitem[Skowron et~al.(2017)Skowron, Lackner, Brill, Peters, and Elkind]{SkowronLaBr16}
Piotr Skowron, Martin Lackner, Markus Brill, Dominik Peters, and Edith Elkind.
\newblock Proportional rankings.
\newblock In \emph{Proceedings of the 26th International Joint Conference on Artificial Intelligence (IJCAI)}, pages 409--415, 2017.

\bibitem[Sorensen et~al.(2024)Sorensen, Moore, Fisher, Gordon, Mireshghallah, Rytting, Ye, Jiang, Lu, Dziri, Althoff, and Choi]{roadmap}
Taylor Sorensen, Jared Moore, Jillian Fisher, Mitchell~L. Gordon, Niloofar Mireshghallah, Christopher~Michael Rytting, Andre Ye, Liwei Jiang, Ximing Lu, Nouha Dziri, Tim Althoff, and Yejin Choi.
\newblock Position: A roadmap to pluralistic alignment.
\newblock In \emph{Proceedings of the 41st International Conference on Machine Learning (ICML)}, pages 46280--46302, 2024.

\bibitem[Viappiani(2015)]{viappiani2015characterization}
Paolo Viappiani.
\newblock Characterization of scoring rules with distances: Application to the clustering of rankings.
\newblock In \emph{Proceedings of the 24th International Joint Conference on Artificial Intelligence (IJCAI)}, pages 104--110, 2015.

\bibitem[Xia and Conitzer(2010)]{xia2010compilation}
Lirong Xia and Vincent Conitzer.
\newblock Compilation complexity of common voting rules.
\newblock In \emph{Proceedings of the 24th AAAI Conference on Artificial Intelligence (AAAI)}, pages 915--920, 2010.

\bibitem[Xiao et~al.(2025)Xiao, Li, Xie, Getzen, Fang, Long, and Su]{xiao2024algorithmic}
Jiancong Xiao, Ziniu Li, Xingyu Xie, Emily Getzen, Cong Fang, Qi~Long, and Weijie~J. Su.
\newblock On the algorithmic bias of aligning large language models with {RLHF}: Preference collapse and matching regularization.
\newblock \emph{Journal of the American Statistical Association}, 120\penalty0 (552):\penalty0 2154--2164, 2025.

\bibitem[Yuan et~al.(2024)Yuan, Di, Zhao, and Naseem]{yuan2024cultural}
Jiahao Yuan, Zixiang Di, Shangzixin Zhao, and Usman Naseem.
\newblock Cultural palette: Pluralising culture alignment via multi-agent palette.
\newblock \emph{arXiv preprint arXiv:2412.11167}, 2024.

\bibitem[Zhang et~al.(2026)Zhang, Milli, Jusko, Smith, Amos, Bouaziz, Revel, Kussman, Sheynin, Titus, et~al.]{zhang2025cultivating}
Lily~Hong Zhang, Smitha Milli, Karen Jusko, Jonathan Smith, Brandon Amos, Wassim Bouaziz, Manon Revel, Jack Kussman, Yasha Sheynin, Lisa Titus, et~al.
\newblock Cultivating pluralism in algorithmic monoculture: The community alignment dataset.
\newblock In \emph{Proceedings of the 14th International Conference on Learning Representations (ICLR)}, 2026.

\end{thebibliography}

\appendix

\section{Proof of \Cref{prop:output-from-input}}
\label{app:output-from-input}

Consider the instance below with $(n, m, k) = (4, 6, 2)$.
\begin{align*}
\pi_1 &= (3, 4, 2, 5, 1, 6); \\
\pi_2 &= (3, 4, 5, 2, 6, 1); \\
\pi_3 &= (4, 3, 2, 5, 6, 1); \\
\pi_4 &= (4, 3, 5, 2, 1, 6).
\end{align*}

Note that this instance is single-peaked\footnote{See the definition in Appendix~\ref{app:proof-restricted-preferences}.} with respect to the ordering $(1, 2, \dots, 6)$.

We first check that positional proportionality cannot be achieved by choosing from the input. 
For each pair of input rankings, there exist $c\in C$ and $\ell\in[m]$ such that candidate $c$ appears in the top $\ell$ ranks \emph{only} within that pair. 
\begin{itemize}
  \item $\pi_1,\pi_2$: candidate $3$ appears in the top $1$ rank.
  \item $\pi_3,\pi_4$: candidate $4$ appears in the top $1$ rank.
  \item $\pi_1,\pi_3$: candidate $2$ appears in the top $3$ ranks.
  \item $\pi_2,\pi_4$: candidate $5$ appears in the top $3$ ranks.
  \item $\pi_1,\pi_4$: candidate $1$ appears in the top $5$ ranks.
  \item $\pi_2,\pi_3$: candidate $6$ appears in the top $5$ ranks.
\end{itemize}
Each pair imposes a positional constraint on the output that requires selecting at least one of those two rankings. Because no output of size two can intersect all six pairs, positional proportionality cannot be achieved.

In a similar vein, we can show that pairwise proportionality cannot be achieved by choosing from the input rankings. Each pair of rankings is uniquely responsible for a pairwise constraint.
\begin{itemize}
  \item $\pi_1,\pi_2$: candidate $3$ is ranked higher than candidate $4$.
  \item $\pi_3,\pi_4$: candidate $4$ is ranked higher than candidate $3$.
  \item $\pi_1,\pi_3$: candidate $2$ is ranked higher than candidate $5$.
  \item $\pi_2,\pi_4$: candidate $5$ is ranked higher than candidate $2$.
  \item $\pi_1,\pi_4$: candidate $1$ is ranked higher than candidate $6$.
  \item $\pi_2,\pi_3$: candidate $6$ is ranked higher than candidate $1$.
\end{itemize}

Finally, we address PSC.
For each pair of input rankings, we identify a solid coalition that is unique to that pair.
\begin{itemize}
  \item $\pi_1,\pi_2$: $\{ 3 \}$.
  \item $\pi_3,\pi_4$: $\{ 4 \}$.
  \item $\pi_1,\pi_3$: $\{ 2, 3, 4 \}$.
  \item $\pi_2,\pi_4$: $\{ 3, 4, 5 \}$.
  \item $\pi_1,\pi_4$: $\{ 1, 2, 3, 4, 5 \}$.
  \item $\pi_2,\pi_3$: $\{ 2, 3, 4, 5, 6 \}$.
\end{itemize}

\section{Proof of \Cref{lem:table-permutation}}
\label{app:omitted-positional-prop}

We proceed by induction on $b$. 
For the base case $b = 1$, the statement holds trivially. 
For the inductive step, assume that for some $b\ge 2$, the statement holds for $b - 1$; we will show that it also holds for~$b$.
Consider any $a\times b$ table as in the lemma statement.
It suffices to show that we can permute each row of the table so that the first column forms a permutation of $[a]$. 
Indeed, this is because in the remaining $a \times (b - 1)$ table that excludes the first column, every integer $c \in [a]$ appears exactly $b - 1$ times, so by invoking the induction hypothesis, we can permute each row of this table so that each column forms a permutation of $[a]$.

To prove our claim, we use Hall's marriage theorem. Construct a bipartite graph $G = (V_L \cup V_R, E)$ where $V_L = \{u_L \mid u \in [a]\}$ represents the integers, $V_R = \{v_R \mid v \in [a]\}$ represents the rows, and an edge $(u_L, v_R) \in E$ exists if and only if row $v$ contains the integer $u$. 
We will show that $G$ has a perfect matching. This perfect matching pairs each integer $c \in [a]$ with a unique row containing it. Consequently, we can permute the entries of each row so that every integer $c$ becomes the first entry of its matched row.

To show that a perfect matching exists, it suffices to check that the marriage condition holds. 
That is, for each set $S \subseteq [a]$ of integers, we need to show that there are at least $|S|$ rows that contain at least one integer from~$S$. 
Since the integers in~$S$ appear exactly $b|S|$ times in the table in total, and each row only contains $b$ integers, these $b|S|$ appearances occur in at least $|S|$ rows.
Hence, the marriage condition is satisfied, as required.

\section{Proof of \Cref{lem:flow-integrality}}
\label{app:omitted-PSC}

We reduce the problem of vertex lower bounds to a standard maximum flow problem with edge capacities. Recall that in a standard flow network with integral capacities, if a valid flow of value $F$ exists, algorithms such as Ford--Fulkerson can find an integral maximum flow in $O(|E| \cdot F)$ time~\citep{CLRS}. The reduction involves two standard techniques: using ``vertex splitting'' to go from vertex constraints to edge constraints, and using a ``super-source'' and ``super-sink'' to convert from lower bounds to upper bounds~\cite{ahuja1993network}.

\paragraph{Auxiliary graph construction.} Let $G = (V, E)$ be the original graph with source $s$, sink $t$, and integral vertex lower bound function $L$. 
Let $f$ be the given valid flow satisfying these constraints. 

We construct an auxiliary directed graph $G' = (V', E')$. For every vertex $v \in V$, we introduce two vertices: an ``in'' vertex $v_{\mathrm{in}}$ and an ``out'' vertex $v_{\mathrm{out}}$. We also add a new source $s^*$ and a new sink~$t^*$. Thus, $V' = \bigcup_{v \in V} \{v_{\mathrm{in}}, v_{\mathrm{out}}\} \cup \{s^*, t^*\}$. 
The edge set $E'$ and the upper capacity function $c$ are defined as follows:
\begin{itemize}[leftmargin=*,nosep]
    \item \emph{Original edges:} For each $(u, v) \in E$, add an edge $(u_{\mathrm{out}}, v_{\mathrm{in}})$ with capacity $c(u_{\mathrm{out}}, v_{\mathrm{in}}) = \infty$.
    \item \emph{Internal edges:} For each $v \in V$, add an edge $(v_{\mathrm{in}}, v_{\mathrm{out}})$ with capacity $c(v_{\mathrm{in}}, v_{\mathrm{out}}) = \infty$.
    \item \emph{Lower bound demands:} For each $v \in V$, add an edge $(s^*, v_{\mathrm{out}})$ with capacity $c(s^*, v_{\mathrm{out}}) = L(v)$, and an edge $(v_{\mathrm{in}}, t^*)$ with capacity $c(v_{\mathrm{in}}, t^*) = L(v)$.
    \item \emph{Circulation edge:} Add a return edge $(t_{\mathrm{out}}, s_{\mathrm{in}})$ from the original sink to the original source, with capacity $c(t_{\mathrm{out}}, s_{\mathrm{in}}) = \ceil{\mathrm{val}(f)}$.
\end{itemize}

\paragraph{Fractional feasibility.} The maximum flow from $s^*$ to $t^*$ is upper-bounded by $\sum_{v \in V} L(v)$, which is the capacity of the cut $(\{s^*\}, V' \setminus \{s^*\})$. 
We show that this maximum is achievable by constructing from $f$ a valid fractional flow $f'$ on $G'$:
\begin{align*}
    f'(u_{\mathrm{out}}, v_{\mathrm{in}}) &= f(u, v) \quad &&\text{for all } (u, v) \in E; \\
    f'(v_{\mathrm{in}}, v_{\mathrm{out}}) &= f(v) - L(v) \quad &&\text{for all } v \in V; \\
    f'(s^*, v_{\mathrm{out}}) &= L(v) \quad &&\text{for all } v \in V; \\
    f'(v_{\mathrm{in}}, t^*) &= L(v) \quad &&\text{for all } v \in V; \\
    f'(t_{\mathrm{out}}, s_{\mathrm{in}}) &= \mathrm{val}(f).
\end{align*}
We have $f'(v_{\mathrm{in}}, v_{\mathrm{out}}) \ge 0$ because the original flow satisfies the vertex lower bounds $f(v) \ge L(v)$. 
It is straightforward to verify that all upper capacity constraints are satisfied: the original and internal edges have infinite capacity, the demand edges connected to $s^*$ and $t^*$ exactly match their capacities, and the flow on the circulation edge satisfies $\mathrm{val}(f) \le \ceil{\mathrm{val}(f)}$. 
Since the edges leaving $s^*$ and entering $t^*$ are saturated, the net flow out of $s^*$ is exactly $\sum_{v \in V} L(v)$. 

We verify flow conservation at the internal vertices of $G'$ in two cases:
\begin{itemize}[leftmargin=*]
    \item \textbf{Vertex $v_{\mathrm{in}}$ for $v \in V$:} 
    If $v \neq s$, the incoming flow comes from the edges $(u_{\mathrm{out}}, v_{\mathrm{in}})$ for $u \in \mathrm{In}(v)$, summing to $f_{\mathrm{in}}(v) = f(v)$. 
    If $v = s$, the incoming flow comes solely from the circulation edge $(t_{\mathrm{out}}, s_{\mathrm{in}})$, which carries $\mathrm{val}(f) = f(s)$. Thus, the incoming flow is always $f(v)$. The outgoing flow splits between $(v_{\mathrm{in}}, t^*)$ carrying $L(v)$ and $(v_{\mathrm{in}}, v_{\mathrm{out}})$ carrying $f(v) - L(v)$, which sums to $f(v)$. 
    Therefore, conservation holds.
    
    \item \textbf{Vertex $v_{\mathrm{out}}$ for $v \in V$:} 
    The incoming flow comes from $s^*$ carrying $L(v)$ and from $v_{\mathrm{in}}$ carrying $f(v) - L(v)$, summing to $f(v)$. For the outgoing flow, if $v \neq t$, it is routed to $(v_{\mathrm{out}}, u_{\mathrm{in}})$ for $u \in \mathrm{Out}(v)$, summing to $f_{\mathrm{out}}(v) = f(v)$. If $v = t$, it is routed solely to the circulation edge $(t_{\mathrm{out}}, s_{\mathrm{in}})$, which carries $\mathrm{val}(f) = f(t)$. 
    Thus, the outgoing flow is always $f(v)$. 
    Therefore, conservation holds.
\end{itemize}

\paragraph{Integrality and verification.} Since a flow of value $\sum_{v \in V} L(v)$ is achievable, and this matches the capacity of the minimum cut $(\{s^*\}, V' \setminus \{s^*\})$, it is a maximum flow. Because all capacities in~$G'$ are integral, the Ford--Fulkerson algorithm is guaranteed to find an integral flow $g^*$ on $G'$ with this same maximum value. 

Since $g^*$ is a maximum flow, it must saturate all edges leaving $s^*$ and entering $t^*$. Thus, $g^*(s^*, v_{\mathrm{out}}) = g^*(v_{\mathrm{in}}, t^*) = L(v)$ for all $v \in V$. 
We translate $g^*$ back to an integral flow $f^*$ on the original graph~$G$ by setting $f^*(u, v) = g^*(u_{\mathrm{out}}, v_{\mathrm{in}})$ for all $(u, v) \in E$. 
We can check that $f^*$ satisfies all required properties:
\begin{enumerate}[label=(\arabic*), leftmargin=*]
    \item \textbf{Flow conservation:} For any $v \in V \setminus \{s, t\}$, the incoming flow is $f^*_{\mathrm{in}}(v) = \sum_{u \in \mathrm{In}(v)} g^*(u_{\mathrm{out}}, v_{\mathrm{in}})$. 
    By conservation at $v_{\mathrm{in}}$ in $G'$, this equals $g^*(v_{\mathrm{in}}, v_{\mathrm{out}}) + g^*(v_{\mathrm{in}}, t^*) = g^*(v_{\mathrm{in}}, v_{\mathrm{out}}) + L(v)$. 
    Similarly, $f^*_{\mathrm{out}}(v) = \sum_{w \in \mathrm{Out}(v)} g^*(v_{\mathrm{out}}, w_{\mathrm{in}})$, which by conservation at $v_{\mathrm{out}}$ in $G'$ equals $g^*(v_{\mathrm{in}}, v_{\mathrm{out}}) + g^*(s^*, v_{\mathrm{out}}) = g^*(v_{\mathrm{in}}, v_{\mathrm{out}}) + L(v)$. 
    Thus, $f^*_{\mathrm{in}}(v) = f^*_{\mathrm{out}}(v)$.
    \item \textbf{Vertex lower bounds:} For any intermediate vertex $v \in V \setminus \{s, t\}$, the flow through $v$ is $f^*(v) = f^*_{\mathrm{in}}(v) = g^*(v_{\mathrm{in}}, v_{\mathrm{out}}) + L(v)$. 
    Since flow values are nonnegative, $f^*(v) \ge L(v)$. 
    For the source~$s$, by a similar reasoning as in the previous paragraph, we have $f^*(s) = f^*_{\mathrm{out}}(s) = g^*(s_{\mathrm{in}}, s_{\mathrm{out}}) + L(s) \ge L(s)$. 
    Similarly, for the sink $t$, we have $f^*(t) = f^*_{\mathrm{in}}(t) = g^*(t_{\mathrm{in}}, t_{\mathrm{out}}) + L(t) \ge L(t)$. 
    Therefore, all demand constraints are satisfied.
    \item \textbf{Value constraint:} We claim that the net value of $f^*$ from $s$ to $t$ is exactly the flow on the circulation edge $(t_{\mathrm{out}}, s_{\mathrm{in}})$ in $G'$. 
    To see this, note that $f^*_{\mathrm{out}}(s) = \sum_{w \in N_{\mathrm{out}}(s)} g^*(s_{\mathrm{out}}, w_{\mathrm{in}})$. 
    By flow conservation at $s_{\mathrm{out}}$ in $G'$, this equals $g^*(s_{\mathrm{in}}, s_{\mathrm{out}}) + g^*(s^*, s_{\mathrm{out}}) = g^*(s_{\mathrm{in}}, s_{\mathrm{out}}) + L(s)$. 
    Similarly, by conservation at $s_{\mathrm{in}}$, the incoming flow $g^*(t_{\mathrm{out}}, s_{\mathrm{in}})$ must equal the outgoing flow $g^*(s_{\mathrm{in}}, s_{\mathrm{out}}) + g^*(s_{\mathrm{in}}, t^*) = g^*(s_{\mathrm{in}}, s_{\mathrm{out}}) + L(s)$. 
    Therefore, $\mathrm{val}(f^*) = f^*_{\mathrm{out}}(s) = g^*(t_{\mathrm{out}}, s_{\mathrm{in}})$. 
    Because this edge has a capacity of $\ceil{\mathrm{val}(f)}$, we have $\mathrm{val}(f^*) \le \ceil{\mathrm{val}(f)}$.
\end{enumerate}

Finally, we analyze the time complexity.
The auxiliary graph $G'$ has $2|V| + 2$ vertices and $|E| + 3|V| + 1$ edges. Running the Ford--Fulkerson algorithm on $G'$ to find $g^*$ takes $O(|E'| \cdot \mathrm{val}(g^*)) = O((|E| + |V|) \sum_{v \in V} L(v))$ time. Translating $g^*$ back into $f^*$ takes $O(|E|)$ time.
This gives the claimed overall time complexity.

\section{Proof of \Cref{thm:pairwise-prop}}
\label{app:proof-pairwise-prop}

We begin by setting up the necessary preliminaries for the general scenario of arbitrary $k \geq 2$ and $n \geq k + 2$.
We first generalize our graph formulation of the problem from the case $k = 2$. 
For each pair of candidates $c, c' \in C$, let $t_{c, c'}$ be the maximum integer such that $c$ is ranked higher than $c'$ in at least $t_{c, c'}\cdot n / k$ input rankings. 
We write $c \nto{t_{c, c'}} c'$ to mean that there is a constraint that $c$ must be ranked higher than $c'$ in at least $t_{c, c'}$ of the $k$ output rankings. 
If $t_{c, c'} = k$, we call such a constraint a \textit{strong constraint}.

We model the constraints as a weighted directed graph whose vertex set is $C = [m]$. 
For each pair $c, c' \in C$, we add an edge $c \to c'$ with weight $t_{c, c'}$.
We will refer to the edges in the graph and the constraints they are induced by interchangeably. 
We next define a family of graphs, which generalize the graph that we used for the case $(n, k) = (4, 2)$.

\begin{definition}
For any integers $r, s \ge 2$, and any $b \in [k - 1]$ and $d \in [s - 1]$, the graph $G(r, s, b, d)$ is defined as a weighted directed graph with vertex set $C = [2rs]$. 
This graph is composed of $s$ cycles $D_j$ for $j \in [s - 1]_0$, where $D_j$ consists of the vertices $u \in \{2rj + 1, 2rj+2, \dots, 2rj + 2r\}$ and, for each $j' \in [r]$, the edges 
\begin{align*}
2rj + 2j' - 1 \nto{b} 2rj + 2j' \quad \text{ and } \quad 2rj + 2j' \nto{k} 2rj + ((2j' + 1) \bmod{2r}).
\end{align*}
For each $j \in [s - 1]_0$, let $E_j$ and $O_j$ be the set of even-numbered and odd-numbered vertices in $D_j$, respectively. 
Taking indices modulo $s$, for each $j \in [s - 1]_0$ and $d' \in [d]$, the graph also has the edge $u \nto{k} v$ for all $u \in E_j$ and $v \in O_{j + d'}$.
\end{definition}

With this notation, the graph that we used earlier for the case $(n, k) = (4, 2)$ is $G(3, 4, 1, 1)$.

Given $G(r, s, b, d)$, we say that a ranking $\pi$ (of the candidates in~$C$) \textit{touches} a cycle $D_j$ if it satisfies a constraint of the type $u \nto{b} v$ in $D_j$. We will prove a useful lemma regarding these graphs.

\begin{lemma}
\label{lem:pairwise_prop_general}
For any integers $r, s \ge 2$, and any $b \in [k - 1]$ and $d \in [s - 1]$, any ranking $\pi$ that satisfies all strong constraints in $G(r, s, b, d)$ touches $D_j$ for at most $s-d$ indices $j \in [s-1]_0$.
\end{lemma}

\begin{proof}
Suppose for contradiction that $\pi$ touches (at least) $s - d + 1$ of the cycles $D_j$. 
Let these cycles be $D_{j_0}, \dots, D_{j_{s - d}}$, where $j_0 < \dots < j_{s - d}$.
For each $q \in [s - d]_0$, let $u_{q} \nto{b} v_{q}$ be a constraint in~$D_{j_{q}}$ that $\pi$ satisfies; in particular, $u_q$ is odd and $v_q$ is even. 
Taking indices modulo $s - d + 1$, we will show that for every $q \in [s - d]_0$, there exists $d' \in [d]$ such that $j_{q + 1} \equiv j_q + d' \pmod s$.
Since $j_0 < \dots < j_{s - d}$, there exists $d'' \in [s-d, s-1]$ such that $j_q \equiv j_{q + 1} + d'' \pmod s$. 
Therefore, $j_{q + 1} \equiv j_q + s - d'' \pmod s$, and we can take $d' = s - d''$.

By definition of $G(r, s, b, d)$, this implies that we have a cycle
\[ u_0 \nto{b} v_0 \nto{k} u_1 \nto{b} v_1 \nto{k} \cdots \nto{k} u_{s - d} \nto{b} v_{s - d} \nto{k} u_0 \]
in which every constraint is satisfied by $\pi$, which is a contradiction.
\end{proof}

We now proceed to prove \Cref{thm:pairwise-prop} in its full generality.

\begin{proof}[Proof of \Cref{thm:pairwise-prop}]
Fix $k \ge 2$ and $n \ge k + 2$.
If $\gcd(k, n) \neq 1$, let $k' = \frac{k}{\gcd(k, n)} \in [k-1]$, so $k'n\equiv 0\pmod k$.
Else, $\gcd(k, n) = 1$; in this case, let $k'$ be the inverse of $-n$ modulo~$k$, so $k'\in[k-1]$ and $k'n\equiv -1\pmod k$.
In either case, let $n' = \ceil{\frac{k'n}{k}}$, which is the number of input rankings in which a candidate $c$ must be ranked higher than another candidate $c'$ in order to have a constraint $c \nto{k'} c'$. 
We have $n' \leq \frac{k'n}{k} + \frac{1}{k}$, which means that $kn' \leq k'n + 1$.
Hence, $kn' + k \leq k'n + k + 1 < k'n + n$, so $k' + 1 > \frac{k(n' + 1)}{n}$ and $\floor{\frac{k(n' + 1)}{n}} \leq k'$.
Since $k' < k$, this implies that $n' + 1 < n$, and thus $n - n' - 1 \in [n - 1]$. 
Therefore, we can consider the graph $G(n' + 1, n, k', n - n' - 1)$. 
This graph has $m = 2n(n' + 1)$ vertices, with each vertex corresponding to a candidate.

We first show that the graph $G(n' + 1, n, k', n - n' - 1)$ cannot be satisfied by any $k$ output rankings. 
For the sake of contradiction, assume otherwise.
By \Cref{lem:pairwise_prop_general}, each output ranking can touch at most $n - (n - n' - 1) = n' + 1$ cycles in the graph. 
By the pigeonhole principle, there is a cycle touched by at most $\floor{\frac{k(n' + 1)}{n}} \le k'$ rankings.
Let $D_j$ be a cycle touched by at most $k'$ output rankings.
If it is touched by at least one output ranking, consider one of these rankings $\pi$. 
The ranking~$\pi$ cannot satisfy every constraint in the cycle simultaneously, so there is a constraint $u \nto{k'} v$ in the cycle which is not satisfied by~$\pi$. 
Note that only rankings that touch $D_j$ can satisfy this constraint, so at most $k' - 1$ output rankings satisfy the constraint.
Else, if $D_j$ is not touched by any output ranking, then since $k' - 1 \ge 0$, any constraint $u \nto{k'} v$ in $D_j$ is again satisfied by at most $k' - 1$ rankings.
In either case, there exists a constraint $u \nto{k'} v$ in $D_j$ that is not satisfied by the $k$~output rankings.

It remains to construct $n$ input rankings that generate the constraint graph $G(n' + 1, n, k', n - n' - 1)$. 
For each $j \in [n - 1]_0$ and $j' \in [2n' + 2]$, let $S_{j, j'}$ be the sequence of vertices that we obtain by traversing the cycle $D_j$ starting from vertex $(2n'+2)j + j'$ and following its edges. Taking indices modulo $n$, for each $j \in [n - 1]_0$, we define\footnote{This is slightly different from the definition that we used for the case $(n,k) = (4,2)$.}
\begin{align*}
\pi_j &= (E_{j}, E_{j + 1}, \dots, E_{j + n - n' - 2}, \\
&\qquad S_{j + n - n' - 1, 2}, S_{j + n - n', 4}, \dots, S_{j + n - 1, 2n' + 2}, \\
&\qquad O_{j}, O_{j + 1}, \dots, O_{j + n - n' - 2}),
\end{align*}
where the elements in each $E_{j''}$ and $O_{j''}$ are arbitrarily ordered.
We take $\pi_0,\pi_1,\dots,\pi_{n-1}$ as our $n$~input rankings. 
Since $n' + 2 \le n$, there is at least one $E_j$ and at least one $O_j$ in the ranking.

We first verify that for each $j \in [n - 1]_0$, the ranking $\pi_j$ satisfies all strong constraints.
\begin{itemize}
\item For all $j' \in [n - 1]_0$ and $a \in [n' + 1]$, the sequence $S_{j', 2a}$ satisfies all strong constraints in $D_{j'}$. 
For each $j' \in [n - 1]_0$, either $E_{j'}$ is ranked higher than $O_{j'}$ in $\pi_j$, or $\pi_j$ contains $S_{j', 2a}$ for some $a\in[n'+1]$, so all strong constraints in $D_{j'}$ are satisfied.
\item For all\footnote{Again, these indices are taken modulo~$n$.} $j' \in [j, j + n - n' - 2]$ and $d' \in [n - n' - 1]$, we have that $E_{j'}$ is ranked higher than either $S_{j' + d', 2a}$ (for some $a\in[n'+1]$) or $O_{j' + d'}$. 
Similarly, for all $j' \in [j + n - n' - 1, j + n - 1]$ and $d' \in [n - n' - 1]$, we have that $S_{j', 2a}$ (for some $a\in[n'+1]$) is ranked higher than either $S_{j' + d', 2a'}$ (for some $a' \in [n' + 1]$) or $O_{j' + d'}$. 
Therefore, the strong constraints across cycles, which are those from $E_{j}$ to $O_{j + d'}$ for $j \in [n - 1]_0$ and $d' \in [n - n' - 1]$, are also satisfied.
\end{itemize}

Finally, observe that for each $j \in [n - 1]_0$, the $n' + 1$ sequences $S_{j, 2}, S_{j, 4}, \dots, S_{j, 2n' + 2}$ collectively satisfy each constraint of the type $u \nto{k'} v$ in $D_j$ exactly $n'$ times, and they appear in $\pi_{j - (n - n' - 1)}, \dots, \pi_{j - (n - 1)}$, respectively.
Hence, all constraints $u \nto{k'} v$ in the graph are also satisfied. 
It follows that these $n$ input rankings generate all constraints in the graph, as desired.
\end{proof}

\section{Proofs of \Cref{thm:single-peaked-simul-prop,thm:pairwise-prop-single-crossing}}
\label{app:proof-restricted-preferences}

For ease of reference, we provide the definitions of the two restricted classes here, and refer to the survey by \citet{ElkindLaPe17} for further discussion.

\begin{definition}[Single-peaked]
\label{def:single-peaked}
A ranking~$\pi$ of the candidates in~$C$ is said to be \emph{single-peaked with respect to $(c_1,\dots,c_m)$}, where $(c_1,\dots,c_m)$ is an ordering of the candidates in~$C$, if there exists $\ell\in[m]$ such that $c_\ell$ is ranked top in~$\pi$, and whenever $1\le \ell'' < \ell' < \ell$ or $\ell < \ell' < \ell'' \le m$, the candidate $c_{\ell'}$ is ranked higher than $c_{\ell''}$ in~$\pi$.
The candidate~$c_\ell$ is called the \emph{peak} of~$\pi$ (or of a voter with the ranking~$\pi$).

A set of rankings is said to be \emph{single-peaked} if there exists an ordering $(c_1,\dots,c_m)$ of the candidates in~$C$ such that each ranking is single-peaked with respect to this ordering.
An instance is called \emph{single-peaked} if its set of input rankings is single-peaked.
\end{definition}

\begin{definition}[Single-crossing]
\label{def:single-crossing}
A set of rankings is said to be \emph{single-crossing with respect to $(i_1,\dots,i_n)$}, where $(i_1,\dots,i_n)$ is an ordering of the voters in~$N$, if for each pair of distinct candidates $c$ and~$c'$, the voters who prefer $c$ to~$c'$ form an interval in the ordering.
A set of rankings is said to be \emph{single-crossing} if it is single-crossing with respect to some ordering of the voters in~$N$.
An instance is called \emph{single-crossing} if its set of input rankings is single-crossing.
\end{definition}

Before starting the proof of \Cref{thm:single-peaked-simul-prop} for single-peaked preferences, we will set up some preliminaries. 
Without loss of generality, assume that the set of input rankings is single-peaked with respect to the ordering $(1, 2, \dots, m)$. 
We start with two simple but useful lemmas. 

\begin{lemma}
\label{lem:single-peaked-prefer-lemma}
For candidates $c < c' < c''$, an input ranking that ranks $c$ higher than $c'$ also ranks $c'$ higher than $c''$ (and therefore ranks $c$ higher than $c''$ as well). 
Similarly, an input ranking that ranks $c''$ higher than $c'$ also ranks $c'$ higher than $c$ (and therefore ranks $c''$ higher than $c$ as well).
\end{lemma}

\begin{proof}
Since $c$ is ranked higher than $c'$, the peak of the ranking must be some candidate $a < c'$. Then, $a < c' < c''$, and so $c'$ is ranked higher than $c''$. 
The other claim can be shown analogously.
\end{proof}

\begin{lemma}\label{lem:single-peaked-prefer-lemma-2}
For candidates $c < c' < c''$, an input ranking that ranks $c$ higher than $c''$ also ranks $c'$ higher than $c''$. Similarly, an input ranking that ranks $c''$ higher than $c$ also ranks $c'$ higher than $c$. 
\end{lemma}

\begin{proof}
Given that $c$ is ranked higher than $c''$, the peak must be a candidate $a < c''$. If $c' < a < c''$, then $c'$ is ranked higher than $c$, and thus $c'$ is ranked higher than $c''$. If $a \leq c'$, then $c'$ is always ranked higher than $c''$. 
The other claim can be proven analogously.
\end{proof}

We will use the same notation for constraints as in the proof of the general case of \Cref{thm:pairwise-prop} (see Appendix~\ref{app:proof-pairwise-prop}). In particular, for candidates $c, c' \in C$, an edge $c \nto{t} c'$ for some $t \in [k]_0$ means that $t$ is the largest integer such that $c$ is ranked higher than $c'$ in at least $t \cdot n/k$ input rankings, and consequently $c$ must be ranked higher than $c'$ in at least $t$ output rankings.

For each $d \in [k]$, let $G_d$ be an unweighted directed graph with vertex set $C = [m]$, such that for each pair of candidates $c < c'$ we have an edge $c \to c'$ if and only if there is a constraint $c \nto{t} c'$ for some $t \in [d,k]$, and for each pair of candidates $c > c'$ we have an edge $c \to c'$ if and only if there is a constraint $c \nto{t} c'$ for some $t \in [k-d+1,k]$. 
For each $d \in [k]$, perform \Cref{alg:topo_sort_gd} on~$G_d$, which outputs a ranking $\pi_d$.

\begin{algorithm}[t]
\caption{Topological Sort of $G_d$}\label{alg:topo_sort_gd}
\begin{algorithmic}[1]
\State \textbf{Input:} Directed graph $G_d$ with vertex set $C = [m]$.
\State \textbf{Output:} Topological ordering of $G_d$ which is single-peaked with respect to the ordering $(1, 2, \dots, m)$. 

\State $a \gets 1, \, b \gets m$
\State $\pi_d \gets (~)$

\While{$a \le b$}
  \If{there does not exist a candidate $c \in [a + 1, b]$ such that $a$ has an outgoing edge to $c$}
    \State Prepend $a$ to the front of $\pi_d$.
    \State $a \gets a + 1$
  \Else
    \State Prepend $b$ to the front of $\pi_d$.
    \State $b \gets b - 1$
  \EndIf
\EndWhile

\State \Return $\pi_d$
\end{algorithmic}
\end{algorithm}

We prove the following proposition about $\pi_d$.
\begin{proposition}
\label{prop:topological-ordering}
For each $d \in [k]$, the ranking $\pi_d$ is a topological ordering of $G_d$ and is single-peaked with respect to the ordering $(1, 2, \dots, m)$.
\end{proposition}

\begin{proof}
Firstly, note that at the start of each iteration, $[a, b]$ represents the set of candidates that have not yet been added to $\pi_d$. Since we always add either $a$ or $b$ to the front of $\pi_d$ and recurse to either $[a + 1, b]$ or $[a, b - 1]$, respectively, the final ranking $\pi_d$ is guaranteed to be single-peaked with respect to the ordering $(1, 2, \dots, m)$.

It remains to show that $\pi_d$ is a topological ordering of $G_d$. Throughout \Cref{alg:topo_sort_gd}, we show that the following invariants always hold:
\begin{enumerate}[label={(\arabic*)}]
    \item No candidate already added to $\pi_d$ has an outgoing edge to any candidate in $[a, b]$.
    \item For candidates $c, c'$ such that $c$ is (not necessarily directly) in front of $c'$ in $\pi_d$, it holds that $c'$ does not have an outgoing edge to $c$.
\end{enumerate}

If these invariants are maintained throughout the algorithm, then invariant~(2) implies that $\pi_d$ is a topological ordering of $G_d$, as $\pi_d$ contains all candidates at the end of the algorithm. 
The invariants trivially hold at the start of the algorithm since $\pi_d$ is empty. 
Consider the two cases for the conditional branch in the algorithm.

\underline{Case 1}: $a$ has no outgoing edge to any candidate in $[a + 1, b]$. 

Invariant (1) implies that if we add $a$ to the front of $\pi_d$, invariant (2) is still preserved. 
Also, by the premise of this case, invariant (1) is still preserved after adding $a$.

\underline{Case 2}: $a$ has an outgoing edge to some candidate $c \in [a + 1, b]$. 

Since we now add $b$ to the front of $\pi_d$ instead of $a$, it suffices to show that $b$ has no outgoing edge to any candidate in $[a, b - 1]$, after which the argument follows as in Case 1.

Suppose for contradiction that $b$ has an outgoing edge to some candidate $c' \in [a, b - 1]$.
Since $c' < b$, by definition of $G_d$, this outgoing edge corresponds to a constraint $b \nto{t} c'$ where $k - d + 1 \leq t \leq k$. 
This means that at least $\frac{tn}{k} \geq \frac{(k - d + 1)n}{k}$ input rankings rank $b$ higher than $c'$. 
As $a \leq c' < b$, by \Cref{lem:single-peaked-prefer-lemma}, at least $\frac{(k - d + 1)n}{k}$ input rankings rank $b$ higher than $a$.
Similarly, by the premise of this case, we can deduce that at least $\frac{dn}{k}$ input rankings rank $a$ higher than $b$.

However, an input ranking cannot rank $a$ higher than $b$ and $b$ higher than $a$ simultaneously, so there are at least $\frac{(k - d + 1)n}{k} + \frac{dn}{k} = \frac{(k + 1)n}{k} > n$ input rankings, a contradiction.
It follows that both invariants are maintained at every step, and $\pi_d$ is indeed a topological ordering of $G_d$.
\end{proof}

We next present a lemma related to positional proportionality.

\begin{lemma}
\label{lem:single-peaked-simul-prop-induct}
Let $c \in C$ be a candidate and $d \in [k]_0$.
Let $\ell \in [m]$ and $t \in [k]$ be such that $c$ appears in the top $\ell$ ranks in at least $t \cdot n/k$ input rankings.
Let $x_d$ be the number of rankings among $\pi_1, \dots, \pi_d$ (which is an empty sequence if $d = 0$) where $c$ appears in the top $\ell$ ranks. Then, at least one of the following holds:
\end{lemma}

\begin{enumerate}[label={(\arabic*)}]
  \item $x_d \geq t$.
  \item There are at least $\frac{(d - x_d + t)n}{k}$ input rankings such that each of these rankings either ranks at least one candidate in\footnote{Since $c-\ell$ may be negative, in which case the set is empty, we avoid using the notation $[c-\ell]$ here.} $\{1,2,\dots,c-\ell\}$ higher than $c$ or ranks $c$ in the top $\ell$ ranks. 
  (It is possible that $c-\ell \le 0$, in which case the set is empty.)
\end{enumerate}

To prove this lemma, we first need to establish another lemma.
\begin{lemma}
\label{lem:single-peaked-helper}
Let $c \in C$ be a candidate and $\ell\in[m]$.
Consider an input ranking that ranks at least one candidate in $\{1,2,\dots,c-\ell\}$ higher than $c$ or ranks $c$ in the top $\ell$ ranks.
Then, this ranking also ranks $c$ higher than every candidate in $\{c  + \ell, c + \ell + 1, \dots, m\}$.
\end{lemma}

\begin{proof}
By the single-peaked property, the candidates in the top $\ell$ ranks in each ranking form a contiguous range $[a, a + \ell - 1]$ for some $a \in [m - \ell + 1]$. 
For rankings in which $c$ appears in the top $\ell$ ranks, we must have $a \leq c$, and so $a + \ell - 1 < c + \ell$. 
Hence, these rankings rank $c$ higher than all candidates in $\{c+\ell,\dots, m\}$, as all such candidates appear outside the top $\ell$ ranks.
On the other hand, by \Cref{lem:single-peaked-prefer-lemma}, rankings that rank at least one candidate in $\{1,\dots,c-\ell\}$ higher than $c$ also rank $c$ higher than all candidates in $\{c+\ell,\dots, m\}$.
\end{proof}

With \Cref{lem:single-peaked-helper} in hand, we can now prove \Cref{lem:single-peaked-simul-prop-induct}.
\begin{proof}[Proof of \Cref{lem:single-peaked-simul-prop-induct}]
We proceed by induction on $d$. For the base case $d = 0$, we have $x_0 = 0$, and so $d - x_d + t = t$. 
Since $c$ appears in the top $\ell$ ranks in at least $t \cdot n/k$ input rankings, the lemma holds for the base case. 
Assume that the inductive hypothesis is true up to some $d \in [k-1]_0$. 
If $x_d \geq t$, then we also have $x_{d + 1} \geq t$, which completes the inductive step for this case.

Suppose from now on that $x_d < t$, so property (2) must be true for~$d$ by the inductive hypothesis. 
By definition, we have $x_{d + 1} \in \{x_d, x_d + 1 \}$. 
If $x_{d + 1} = x_d + 1$, then $d + 1 - x_{d + 1} + t = d - x_d + t$, so property (2) is still true for~$d+1$. 
Otherwise, $x_{d + 1} = x_d$, and so $c$ is not in the top $\ell$ ranks of $\pi_{d + 1}$. 
Suppose that in the iteration of \Cref{alg:topo_sort_gd} when $c$ is added to $\pi_{d + 1}$, the unranked candidates form a range $[a, b]$. 
Since $c$ is added at this step, we have $c = a$ or $c = b$.

\underline{Case 1}: $c = a$.

As $c$ is not in the top $\ell$ ranks of $\pi_{d + 1}$, we have $b \geq c + \ell$. 
Since $x_d < t$, it holds that $d - x_d + t \geq d + 1$, so by \Cref{lem:single-peaked-helper}, for each $c'' \in \{c + \ell, \dots, m\}$, there exists $t' \in [d + 1, k]$ such that there is a constraint $c \nto{t'} c''$. 
By definition of $G_{d + 1}$, there is an edge $c \to c''$ for each $c'' \in \{c + \ell, \dots, m\}$. 
As $c$ cannot have an outgoing edge to any candidate in $[c + 1, b]$ when we add $c$ to the topological ordering, it holds that $b \leq c + \ell - 1$, which contradicts $b \ge c + \ell$.
Hence, this case cannot occur.

\underline{Case 2}: $c = b$.

As $c$ is not in the top $\ell$ ranks of $\pi_{d + 1}$, we have $a \leq c - \ell$. 
By our algorithm, we only choose $b$ as the next candidate to be added if $a$ has an outgoing edge to some candidate in $[a + 1, b]$. 
Thus, by definition of $G_{d + 1}$, there is a constraint $a \nto{t'} c'$ for some $c' \in [a + 1, b]$ and $t' \in [d+1, k]$. 
By \Cref{lem:single-peaked-prefer-lemma}, all input rankings that rank $a$ higher than $c'$ also rank $a$ higher than $b = c$, so at least $\frac{(d + 1)n}{k}$ input rankings rank $a$ higher than $c$.

Note that input rankings that rank $c$ in the top $\ell$ positions cannot rank $a$ higher than $c$. 
Indeed, the top $\ell$ positions must form a contiguous range of $\ell$ candidates containing $c$, and $a$ is outside this range since $a \le c - \ell$.
Thus, there are at least $\frac{(d + 1)n}{k} + \frac{tn}{k} = \frac{(d + 1 + t)n}{k}$ input rankings that rank $a$ higher than $c$ or rank $c$ in the top $\ell$ positions.
Since $d + 1 + t \geq d + 1 - x_{d + 1} + t$, property~(2) still holds for $d + 1$. 
This completes the induction.
\end{proof}

Next, we introduce a lemma related to PSC, which is analogous to \Cref{lem:single-peaked-simul-prop-induct} for positional proportionality.

\begin{lemma}
\label{lem:single-peaked-simul-prop-induct-2}
Let $C' = [a, b]$ for some $a, b \in C$ with $a \leq b$, and $d \in [k]_0$.
Let $t \in [k]$ be such that all candidates in $C'$ appear in the top $|C'|$ ranks in at least $t \cdot n/k$ input rankings.
Let $x_d$ be the number of rankings among $\pi_1, \dots, \pi_d$ (which is an empty sequence if $d = 0$) where all candidates in $C'$ appear in the top $|C'|$ ranks. Then, at least one of the following holds:
\end{lemma}

\begin{enumerate}[label={(\arabic*)}]
  \item $x_d \geq t$.
  \item There are at least $\frac{(d - x_d + t)n}{k}$ input rankings such that each of these rankings either ranks at least one candidate in $\{ 1, 2, \dots, a - 1 \}$ higher than $b$ or ranks all candidates in $C'$ in the top $|C'|$ ranks.
\end{enumerate}

To prove this lemma, we again need to establish another lemma.

\begin{lemma}
\label{lem:single-peaked-helper-2}
Let $C' = [a, b]$ for some $a, b \in C$ with $a \leq b$. 
Consider an input ranking that ranks at least one candidate in $\{ 1, 2, \dots, a - 1 \}$ higher than $b$ or ranks all candidates in $C'$ in the top $|C'|$ ranks.
Then, this ranking also ranks $a$ higher than all candidates in $\{b + 1, \dots, m\}$.
\end{lemma}

\begin{proof}
If an input ranking ranks all candidates in $C'$ in the top $|C'|$ ranks, then all other candidates are outside the top $|C'|$ ranks, and so all candidates in $C'$ (including $a$) are ranked higher than all candidates in $\{b + 1, \dots, m\}$.

On the other hand, if an input ranking ranks some candidate $c \in \{ 1, 2, \dots, a - 1 \}$ higher than $b$, by Lemma \ref{lem:single-peaked-prefer-lemma}, it ranks $b$ higher than all candidates in $\{b + 1, \dots, m\}$. 
Since $c < b$ and $c$ is ranked higher than $b$, by \Cref{lem:single-peaked-prefer-lemma-2}, all candidates in $C'\setminus\{b\}$ are ranked higher than $b$.
It follows that $a\in C'$ is ranked higher than all candidates in $\{b + 1, \dots, m\}$.
\end{proof}

With Lemma \ref{lem:single-peaked-helper-2} in hand, we can now prove Lemma \ref{lem:single-peaked-simul-prop-induct-2}.

\begin{proof}
We proceed by induction on $d$. For the base case $d = 0$, we have $x_0 = 0$, and so $d - x_d + t = t$. 
Since all candidates $C'$ appear in the top $|C'|$ ranks in at least $t \cdot n/k$ input rankings, the lemma holds for the base case. 
Assume that the inductive hypothesis is true up to some $d \in [k-1]_0$. 
If $x_d \geq t$, then we also have $x_{d + 1} \geq t$, which completes the inductive step for this case.

Suppose from now on that $x_d < t$, so property (2) must be true for~$d$ by the inductive hypothesis. 
By definition, we have $x_{d + 1} \in \{x_d, x_d + 1 \}$. 
If $x_{d + 1} = x_d + 1$, then $d + 1 - x_{d + 1} + t = d - x_d + t$, so property (2) is still true for~$d+1$.  
Otherwise, $x_{d + 1} = x_d$, and so not all candidates in $C'$ appear in the top $|C'|$ ranks of $\pi_{d + 1}$. Observe that by \Cref{alg:topo_sort_gd}, the candidates in the top $|C'|$ ranks of $\pi_{d + 1}$ must form an interval $[a', b']$ such that $a' \neq a$ and $b' - a' = b - a$.

\underline{Case 1}: $a' > a$.

Since $b' - a' = b - a$, we have $b' > b$.
When $a$ is added to $\pi_{d + 1}$, the unranked candidates must form an interval $[a, b'']$ for some $b'' \in [b', m]$. 
This means that there is no outgoing edge from $a$ to any candidate in $[a + 1, b'']$.

Since $x_d < t$, it holds that $d - x_d + t \geq d + 1$, so by Lemma \ref{lem:single-peaked-helper-2}, for each $c \in \{ b + 1, \dots, m \}$, there exists $t' \in [d + 1, k]$ such that there is a constraint $a \nto{t'} c$. 
However, $b'' \in \{ b + 1, \dots, m \}$, so by definition of $G_{d + 1}$, there is an outgoing edge from $a$ to $b''$, a contradiction.
Hence, this case cannot occur.

\underline{Case 2}: $a' < a$.

Since $b' - a' = b - a$, we have $b' < b$.
When $b$ is added to $\pi_{d + 1}$, the unranked candidates must form an interval $[a'', b]$ for some $a'' \in [a']$. 
Since $b$ is added at this step, $a''$ must have an outgoing edge to some candidate in $[a'' + 1, b]$. By definition of $G_{d + 1}$, there is a constraint $a'' \nto{t'} c$ for some $c \in [a'' + 1, b]$ and $t' \in [d + 1, k]$, which implies that at least $\frac{(d + 1)n}{k}$ input rankings rank $a''$ higher than $c$. By Lemma \ref{lem:single-peaked-prefer-lemma}, at least $\frac{(d + 1)n}{k}$ input rankings rank $a''$ higher than $b$. 

Note that input rankings that rank $a''$ higher than $b$ cannot rank all candidates in $C'$ in the top $|C'|$ ranks, since $a'' \notin C'$ and $b \in C'$. Thus, there are at least $\frac{(d + 1)n}{k} + \frac{tn}{k} = \frac{(d + 1 + t)n}{k}$ input rankings that rank $a''$ higher than $b$ or rank all candidates in $C'$ in the top $|C'|$ ranks. Since $d + 1 + t \geq d + 1 - x_{d + 1} + t$, property~(2) still holds for $d + 1$. 
This completes the induction.
\end{proof}

Finally, we are ready to prove \Cref{thm:single-peaked-simul-prop}.
\begin{proof}[Proof of \Cref{thm:single-peaked-simul-prop}]
Let the $k$ output rankings be $\pi_1, \pi_2, \dots, \pi_k$. As shown in \Cref{prop:topological-ordering}, these rankings are single-peaked with respect to the ordering $(1, 2, \dots, m)$. 

First, we verify that this output satisfies pairwise proportionality. 
If we have a constraint $c \nto{t} c'$ for some $c < c'$ and $t \in [k]$, then the graphs $G_1, \dots, G_t$ contain the edge $c \to c'$, and thus $c$ is ranked higher than $c'$ in the $t$ rankings $\pi_1, \dots, \pi_t$ since they are topological orderings of $G_1, \dots, G_t$, respectively. 
Similarly, if we have a constraint $c \nto{t} c'$ for some $c > c'$ and $t \in [k]$, then the graphs $G_{k - t + 1}, \dots, G_k$ contain the edge $c \to c'$, and thus $c$ is ranked higher than $c'$ in the $t$ rankings $\pi_{k - t + 1}, \dots, \pi_k$.

Next, we verify that this output satisfies positional proportionality.
Suppose that $c$ appears in the top $\ell$ ranks in at least $t \cdot n/k$ input rankings, for some $\ell \in [m]$ and $t \in [k]$. By \Cref{lem:single-peaked-simul-prop-induct}, using the same definition of $x_k$, if $x_k < t$, then there are at least $\frac{(k - x_k + t)n}{k} \geq \frac{(k + 1)n}{k} > n$ input rankings, which is clearly false. 
Hence, $x_k \geq t$, which precisely means that $c$ appears in the top $\ell$ ranks in at least $t$ of $\pi_1, \dots, \pi_k$.

We can show that this output also satisfies PSC by using \Cref{lem:single-peaked-simul-prop-induct-2} analogously.
Indeed, observe that for each input ranking and any $\ell \in [m]$, the candidates in the top $\ell$ positions form an interval $C' = [a, b]$ for some $a, b \in C$ with $b = a + \ell - 1$, due to the single-peaked property.
Hence, \Cref{lem:single-peaked-simul-prop-induct-2} can be applied to any solid coalition of voters.

It remains to show that we can construct $\pi_1, \dots, \pi_k$ in polynomial time. 
For each pair of candidates $c$ and $c'$, let $t_{c, c'}$ be the maximum integer such that $c$ is ranked higher than $c'$ in at least $t_{c, c'}\cdot n/k$ input rankings. 
We can compute $t_{c, c'}$ for all $c,c'\in C$ in $O(nm^2)$ time, which allows us to construct the adjacency matrices of the graphs $G_1, \dots, G_k$ in $O(m^2k)$ time.
Each run of \Cref{alg:topo_sort_gd} takes $O(m^2)$ time, and since we run it for all $d \in [k]$, the runs take $O(m^2k)$ time in total. 
It follows that we can output $\pi_1, \dots, \pi_k$ in $O(m^2(n + k))$ time.
\end{proof}

Next, we prove \Cref{thm:pairwise-prop-single-crossing} for single-crossing preferences.
We need the following technical lemma.

\begin{lemma}
\label{lem:interval}
Let $a,b\in[n]$ be such that $b \ge a + \ceil{t \cdot n/k} - 1$ for some $t\in[k]$.
Then, the interval $[a,b]$ contains $\ceil{d \cdot n/k}$ for at least $t$ values of $d\in[k]$.
\end{lemma}

\begin{proof}
Observe that for each $r\in[n]$, the set $[r]$ contains $\ceil{d\cdot n/k}$ for exactly $\floor{rk/n}$ values of $d\in[k]$.
Therefore, the interval $[a,b]$ contains $\ceil{d \cdot n/k}$ for $\floor{bk/n} - \floor{(a-1)k/n} \ge \floor{(b-a+1)k/n} \ge t$ values of $d\in [k]$, where we use the fact that $\floor{y+z} \ge \floor{y} + \floor{z}$ for any real numbers $y,z$.
\end{proof}

We proceed to the proof of \Cref{thm:pairwise-prop-single-crossing}.

\begin{proof}[Proof of \Cref{thm:pairwise-prop-single-crossing}]
Without loss of generality, assume that the set of input rankings is single-crossing with respect to the ordering $(1, 2, \dots, n)$. For each $d \in [k]$, let $\sigma_d$ be the ranking $\pi_{\ceil{d\cdot n/k}}$ submitted by voter $\ceil{d \cdot n/k}$. 
We let $\sigma_1, \dots, \sigma_k$ be the $k$ output rankings.    

From the single-crossing definition, for any pair of candidates $c,c' \in C$, the set of voters $N' \subseteq N$ who prefer $c$ to $c'$ forms either a prefix or a suffix of $(1, 2, \dots, n)$. 
Similarly, for any subset of candidates $C'\subseteq C$, the set of voters who rank the candidates from~$C'$ in the top $|C'|$~ranks forms an interval of $(1, 2, \dots, n)$ (but not necessarily a prefix or suffix).
Indeed, if two voters rank the candidates from~$C'$ in the top $|C'|$~ranks, then so does every voter between them.

We will establish pairwise proportionality; the proof for PSC is analogous.
Suppose that for some $t\in [k]$, at least $t \cdot n/k$ voters prefer candidate~$c$ to candidate~$c'$.
These voters must form an interval $[a,b]$ of $(1, 2, \dots, n)$, where $b \ge a + \ceil{t \cdot n/k} - 1$.
By \Cref{lem:interval}, the rankings of at least $t$ of these voters are chosen in the output, which means that pairwise proportionality is satisfied.

Finally, we can trivially output the rankings in $O(mk)$ time by computing $\ceil{d \cdot n/k}$ for each $d \in [k]$ and outputting the corresponding input ranking.
\end{proof}

It is worth noting that the other positive results in this paper \emph{cannot} be achieved by selecting the output solely from the input rankings.
Moreover, for single-crossing instances, this output restriction also renders positional proportionality unattainable.
We refer to Appendices~\ref{app:output-from-input} and \ref{app:single-crossing-output-from-input} for detailed discussions.

\section{Choosing output from input rankings for single-crossing instances}
\label{app:single-crossing-output-from-input}

For single-crossing instances, one may hope that the approach used in \Cref{thm:pairwise-prop-single-crossing} to achieve pairwise proportionality and PSC could also be used to achieve positional proportionality.
We show not only that this approach fails, but also that it is fundamentally impossible to attain positional proportionality for single-crossing instances by choosing from the input.
Consider the instance below with $(n,m,k) = (4, 14, 2)$.
\begin{align*}
\pi_1 &= (1, 2, 3, 4, 5, 6, 7, 8, 9, 10, 11, 12, 13, 14); \\
\pi_2 &= (1, 3, 2, 4, 5, 7, 6, 8, 10, 9, 11, 12, 14, 13); \\
\pi_3 &= (3, 2, 1, 4, 5, 8, 7, 6, 10, 9, 11, 12, 14, 13); \\
\pi_4 &= (3, 4, 2, 1, 8, 7, 5, 6, 11, 10, 9, 14, 13, 12).
\end{align*}
Intuitively, this instance is composed of four non-crossing sets of candidates: $\{1,2,3,4\}$, $\{5,6,7,8\}$, $\{9,10,11\}$, $\{12,13,14\}$.
One can verify that the instance is indeed single-crossing by checking the pairwise relations within each set of candidates.
Now, for each pair of input rankings, we show that there exist $c\in C$ and $\ell\in[m]$ such that candidate $c$ appears in the top $\ell$ ranks \emph{only} within that pair.
\begin{itemize}
  \item $\pi_1,\pi_2$: candidate $1$ appears in the top $1$ rank.
  \item $\pi_3,\pi_4$: candidate $3$ appears in the top $1$ rank.
  \item $\pi_1,\pi_3$: candidate $2$ appears in the top $2$ ranks.
  \item $\pi_2,\pi_4$: candidate $7$ appears in the top $6$ ranks.
  \item $\pi_1,\pi_4$: candidate $13$ appears in the top $13$ ranks.
  \item $\pi_2,\pi_3$: candidate $10$ appears in the top $9$ ranks.
\end{itemize}
Each pair imposes a positional constraint on the output that requires selecting at least one of those two rankings. 
Because no output of size two can intersect all six pairs, positional proportionality cannot be achieved.

\end{document}